\documentclass[runningheads]{llncs}

\usepackage{amsmath}
\usepackage{float}
\usepackage{times}
\usepackage[utf8]{inputenc}
\usepackage[dvipsnames]{xcolor}
\usepackage{url,mathtools,amssymb,tikz,amsmath}
\usepackage{caption}
\usepackage{subcaption}
\usepackage{graphicx}
\usepackage{multicol}
\usepackage{enumerate}
\usepackage{thm-restate}
\usepackage{wrapfig}

\usepackage[disable]{todonotes} 
\usepackage{nicefrac} 
\usepackage{bbm} 
\usepackage[numbers,sort&compress]{natbib}

\usepackage{hyperref}
\hypersetup{
	colorlinks=true
}

\DeclarePairedDelimiter{\ceil}{\lceil}{\rceil}
\allowdisplaybreaks
\usetikzlibrary{automata,arrows,calc}
\tikzset{->,>=stealth',auto,node distance=2cm,thick,initial text=}
\tikzstyle{accepting}=[path picture={%
  \draw let 
    \p1 = (path picture bounding box.east),
    \p2 = (path picture bounding box.center)
    in
    (\p2) circle (\x1 - \x2 - 2pt);
 }]
\definecolor{green1}{rgb}{0, 0.5, 0}
\definecolor{red1}{rgb}{0.64, 0, 0}

\newcommand{\R}{\mathbb{R}} 
\newcommand{\IfunM}{\Inter_{min}} 
\newcommand{\mx}[1]{max^{#1}} 
\newcommand{\mn}[1]{min^{#1}} 
\newcommand{\Path}[2]{\infl{#1}{#2}} 
\newcommand{\CBfactor}{\beta}
\newcommand{\CBfun}[3]{\CBfactor^{#3}_{\agent{#1},\agent{#2}}} 
\newcommand{\CBfunM}{\CBfactor_{min}}

\newcommand{\Agents}{\mathcal{A}} 
\newcommand{\agent}[1]{#1} 
\newcommand{\Blf}{B} 
\newcommand{\Blft}[1]{\Blf^{#1}} 
\newcommand{\Bfun}[2]{\Blf^{#2}_{\agent{#1}}} 

\newcommand{\qm}[1]{``#1''}

\newcommand{\cali}{\mathcal{I}}

\newcommand{\caly}{\mathcal{Y}}
\newcommand{\calt}{\mathcal{T}}

\newcommand{\Pol}{\rho}
\newcommand{\Pfun}[1]{\Pol(#1)}
\newcommand{\PolER}{\rho_{\mathit{ER}}}
\newcommand{\PfunER}[1]{\PolER(#1)}

\newcommand{\Inter}{\cali} 
\newcommand{\Interclique}{\Inter^{\textit{clique}}} 
\newcommand{\Interdisconnected}{\Inter^{\textit{disc}}} 
\newcommand{\Interunrelenting}{\Inter^{\textit{unrel}}} 
\newcommand{\Intercircular}{\Inter^{\textit{circ}}} 

\newcommand{\Ifun}[2]{\Inter_{#1,#2}} 
\newcommand{\Ifunclique}[2]{\Interclique_{#1,#2}} 
\newcommand{\Ifundisconnected}[2]{\Interdisconnected_{#1,#2}} 
\newcommand{\Ifununrelenting}[2]{\Interunrelenting_{#1,#2}} 
\newcommand{\Ifuncircular}[2]{\Intercircular_{#1,#2}} 

\newcommand{\tmax}{T} 

\newcommand{\Upd}{\mu} 

\newcommand{\UpdR}{\mu^{C}} 

\newcommand{\UpdCB}{\mu^{\textit{CB}}} 

\newcommand{\larrow}[1]{\stackrel{\,\, \small #1\,\,}{\rightarrow}} 
\newcommand{\Larrow}[2]{\stackrel{\,\,#1\,\,}{\leadsto_{#2}}} 
\newcommand{\ldinfl}[3]{#1{\larrow{\begin{tiny}{#2}\end{tiny}}}#3} 
\newcommand{\infl}[2]{#1{\Larrow{}{}}#2} 
\newcommand{\linfl}[4]{#1 \Larrow{#2}{#3} #4} 

\newcommand{\nat}{\mathbb{N}}
\newcommand{\reals}{\mathbb{R}}
\newcommand{\mstar}{\mathtt{m}}

\newcommand{\defsymbol}{\stackrel{\textup{\texttt{def}}}  {=}} 

\def\techrep{} 

\ifdefined\techrep
 \newcommand{\version}[2]{#2} 
\else
 \newcommand{\version}[2]{#1} 
\fi

\begin{document}
\title{A Multi-Agent Model for Polarization under Confirmation Bias in Social Networks
\thanks{M\'ario S. Alvim and Bernardo Amorim were partially supported by CNPq, CAPES and FAPEMIG.
Santiago Quintero and Frank Valencia were partially supported by the ECOS-NORD project FACTS (C19M03).}}

\titlerunning{A Multi-Agent Model for Polarization under Confirmation Bias in Social Networks}
%
\author{M\'{a}rio S. Alvim \inst{1}
\and
Bernardo Amorim\inst{1}
\and
Sophia Knight\inst{2}
\and
Santiago Quintero\inst{3}
\and
Frank Valencia\inst{4,5}
}
\authorrunning{Alvim et al.}
%
\institute{
Department of Computer Science, UFMG, Brazil \and
Department of Computer Science, University of Minnesota Duluth, USA
\and
LIX, \'{E}cole Polytechnique de Paris, France \and
CNRS-LIX, \'{E}cole Polytechnique de Paris, France \and
Pontificia Universidad Javeriana Cali, Colombia
}
\maketitle              
\begin{abstract}
\todo{Mario on 2021-02-13: I put the paper in the LNCS format for FORTE.}
We describe a model for polarization in multi-agent systems based on Esteban and Ray's standard measure of polarization from economics. Agents evolve by updating their beliefs (opinions) based on an underlying influence graph, as in the standard DeGroot model for social learning, but under a \emph{confirmation bias}; i.e., a discounting of opinions of agents with dissimilar views. 
We show that even under this bias polarization eventually vanishes (converges to zero) if the influence graph is strongly-connected. If the influence graph is a regular symmetric circulation, we determine the unique belief value to which all agents converge. Our more insightful result establishes that, under some natural assumptions, if polarization does not eventually vanish then either there is a disconnected subgroup of agents, or some agent influences others more than she is influenced.
We also show that polarization does not necessarily vanish in weakly-connected graphs under confirmation bias.
We illustrate our model with a series of case studies 
and simulations, and show how it relates to the classic DeGroot model 
for social learning.

\end{abstract}

\section{Introduction}
\label{sec:introduction}


\todo{Mario on 2020-10-19: Add this to intro somewhere? ``A central concern of our model is that the majority of the (quite extensive) literature on the representation of agents' beliefs --as well as of the way these agents interact to update such beliefs-- assumes that agents are rational and optimally use all information at their disposal. But a growing body of work by psychologists, sociologists, and economists has challenged these assumptions, and several common, relevant cognitive biases  have been catalogued~\cite{XXX,YYY,ZZZ}. In particular, we model the behavior of agents prone to \emph{authority bias}, by which one gives more weight to evidence presented by some agents than by others, and to \emph{confirmation bias}, by which one tends to give more weight to evidence supporting their current beliefs than to  evidence contradicting them, independently from whence the evidence is coming.''}

\emph{Distributed systems} have changed substantially in the recent past with the
advent of social networks.  In the
previous incarnation of distributed computing~\cite{Lynch96} the
emphasis was on consistency, fault tolerance, resource management and
related topics; these were all characterized by \emph{interaction between
processes}.  What marks the new era of distributed systems is an emphasis on the flow of epistemic information (facts, beliefs, lies) and its impact on  democracy and on society at large.


Indeed in social networks a group may shape their beliefs by attributing more value to the opinions of outside influential figures. This cognitive bias is known as \emph{authority bias} \cite{Ramos:19:Book}. Furthermore, in a group with uniform views, users may become extreme by reinforcing one another's opinions, giving more value to opinions that confirm their own preexisting beliefs. This is another common cognitive bias known as \emph{confirmation bias}~\cite{Aronson10}. As a result, social networks can cause their users to become radical and isolated in their own ideological circle causing dangerous splits in society~\cite{Bozdag13} in a phenomenon known as \emph{polarization}~\cite{Aronson10}.

There is a growing interest in the development of models for the analysis of polarization and social influence in networks \cite{li,proskurnikov,sirbu,gargiulo,alexis,Guerra,myp,degroot,naive,zoe,fblogic,facebook,hunter}. Since polarization involves non-terminating systems with \emph{multiple agents} simultaneously exchanging information (opinions), concurrency models are a natural choice to capture the dynamics of polarization.  

\todo{Mario to Frank on 2020-10-19: Do we want to refer to ourselves as ``the authors'' instead of ``we''?}
\emph{The Model.} In fact, we developed a multi-agent model for polarization in \cite{Alvim:19:FC}, inspired by linear-time models of concurrency where the state of the system evolves in discrete time units (in particular \cite{tcc,ntcc}). In each time unit, the agents \emph{update} their beliefs  about the proposition of interest taking into account the beliefs of their neighbors in an underlying weighted \emph{influence graph}. The belief update gives more value to the opinion of agents with higher influence (\emph{authority bias}) and to the opinion of agents with similar views (\emph{confirmation bias}).  Furthermore, the model is equipped with  a \emph{polarization measure} based on  the seminal work in economics by Esteban and Ray~\cite{Esteban:94:Econometrica}.  The polarization is measured at each time unit and it is $0$ if all agents' beliefs fall within an interval of agreement about the proposition. The contributions in \cite{Alvim:19:FC} were of an experimental nature and aimed at exploring how the combination of influence graphs and cognitive biases in our model can lead to polarization.

In the current paper we prove claims made from experimental observations in \cite{Alvim:19:FC} using techniques from calculus, graph theory, and flow networks. The main goal of this paper is identifying how networks and beliefs are structured, for agents subject to confirmation bias, when polarization \emph{does not} disappear. Our results provide insight into the phenomenon of polarization, and are a step toward the design of robust computational models and simulation software for human cognitive and social processes. 

The closest related work is that on DeGroot models \cite{degroot}. These are the standard linear models for social learning whose analysis can be carried out by linear techniques from Markov chains. A novelty in our model is that its update function extends the classical update  from DeGroot models with confirmation bias. As we shall elaborate in Section~\ref{sec:degroot} the extension makes the model no longer linear and thus mathematical tools like Markov chains do not seem applicable. Our model incorporates a polarization measure in a model for social learning and extends classical convergence results of DeGroot models to the confirmation bias case.

\emph{Main Contributions.}  The following are the main theoretical results established in this paper. Assuming confirmation bias and some natural conditions about belief values: (1)  {If polarization does not disappear then either there is disconnected subgroup of agents, or some agent influences others more than she is influenced, or all the agents are initially radicalized} (i.e., each individual holds the most extreme value either in favor or against of a given proposition). (2) {Polarization eventually disappears  (converges to zero) if the influence graph is strongly-connected}. (3) {If the influence graph is a regular symmetric circulation we determine the unique belief value all agents converge to.}


\emph{Organization.} In Section \ref{sec:model} we introduce the model and  illustrate a series of examples and simulations, uncovering interesting new insights and complex characteristics of the believe evolution. The theoretical contributions (1-3) above are given in Sections \ref{sec:general-result} and \ref{sec:specific-cases}. We discuss DeGroot and other related work in Sections \ref{sec:degroot}  and \ref{sec:conclusion}. 
\version{Full proofs can be found in the corresponding technical report~\cite{Alvim:21:ForteTechRep}.}{Full proofs are in the Appendix.}
An implementation of the model in Python and the simulations are available on Github~\cite{website:github-repo}.


\section{The Model}
\label{sec:model}
Here we refine the polarization model introduced in \cite{Alvim:19:FC}, composed of {static} and {dynamic} elements.
We presuppose basic knowledge of calculus and graph theory \cite{Sohrab:14,Diestel:17}.
\subsubsection{Static Elements of the Model}
\label{sec:model-static}

\emph{Static elements} of the model represent a snapshot of a social network 
at a given point in time. They
include the following components:

\begin{itemize}
 \item A (finite) set $\Agents = \{\agent{0}, \agent{1}, \ldots, \agent{n{-}1} \}$ of $n  \geq 1$ \emph{agents}.
 
 \item A \emph{proposition} $p$ of interest, about which agents can hold beliefs.
 
 \item A \emph{belief configuration}    
    $\Blf{:}\Agents{\rightarrow}[0,1]$ s.t.
    each value $\Blf_{\agent{i}}$ is the instantaneous confidence of agent
    $\agent{i}{\in}\Agents$ in the veracity of proposition $p$.
    Extreme values $0$ and $1$ represent a firm belief in, respectively, the falsehood or truth of  $p$.
 
 \item A \emph{polarization measure} $\Pol{:}[0,1]^{\Agents}{\rightarrow}\R$ 
 mapping belief configurations to real numbers.
 The value $\Pfun{\Blf}$ indicates how polarized belief configuration 
 $\Blf$ is.
\end{itemize}

There are several polarization measures described in the literature.
In this work we adopt the influential measure proposed by 
Esteban and Ray~\cite{Esteban:94:Econometrica}.

\begin{definition}[Esteban-Ray Polarization]
\label{def:poler}
Consider a set $\caly{=}\{y_0, \allowbreak y_1, \allowbreak \ldots, \allowbreak y_{k-1}\}$ of size $k$, s.t.\ each $y_i{\in}\mathbb{R}$.
Let $(\pi, y){=}(\pi_0, \allowbreak \pi_1, \allowbreak \ldots, \allowbreak \pi_{k{-}1}, \allowbreak y_0, \allowbreak y_1, \allowbreak \ldots, \allowbreak y_{k{-}1})$ 
be a \emph{distribution} on $\caly$ s.t.\ $\pi_i$ is the frequency of value $y_i{\in}\caly$ 
in the distribution. 
\footnote{W.l.o.g. we can assume the values of $\pi_i$ are all non-zero and add up to 1.
} 
The \emph{Esteban-Ray (ER) polarization measure} is defined as
$
    \PfunER{\pi, y} = K \sum_{i=0}^{k-1} \sum_{j=0}^{k-1} \pi_i^{1+\alpha} \pi_j | y_i - y_j |, 
$
where $K{>}0$ is a constant, and typically $\alpha{\approx}1.6$.
\end{definition}

The higher the value of $\PfunER{\pi, y}$, the more polarized  
distribution $(\pi,y)$ is.
The measure captures the intuition that 
polarization is accentuated by both intra-group homogeneity
and inter-group heterogeneity.
Moreover, it assumes that the total polarization 
is the sum of the effects of individual agents on one another.
\version{The measure can be derived from a set of intuitively reasonable axioms~\cite{Esteban:94:Econometrica}.}{The measure can be derived from a set of intuitively reasonable axioms~\cite{Esteban:94:Econometrica}, described in 
Appendix~\ref{sec:polar-axioms}.}

Note that $\PolER$ is defined on a discrete distribution, 
whereas in our model a general polarization metric is defined on 
a belief configuration $\Blf{:}\Agents{\rightarrow}[0,1]$. 
To apply $\PolER$ to our setup we convert the belief configuration 
$\Blf$ into an appropriate distribution $(\pi,y)$.

\begin{definition}[$k$-bin polarization]\label{k-bin:def}
Let $D_k$ be a discretization of the interval $[0,1]$  into 
$k{>}0$ consecutive non-overlapping, non-empty intervals (\emph{bins})  $I_0,I_1,\ldots, I_{k-1}$.
We use the term \emph{borderline  points} of $D_k$ to refer to the end-points  of $I_0,I_1,\ldots, I_{k-1}$  different from 0 and 1. We assume an underlying discretization $D_k$ throughout the paper. 

Given $D_k$ and a belief configuration $B$, define the distribution $(\pi, y)$ as follows. Let $\caly{=}\{y_0,y_1,\ldots,y_{k-1}\}$ where each $y_i$ is the mid-point of $I_i$,
and let  $\pi_i$ be the fraction of agents having their belief in $I_{i}.$
 The polarization measure $\Pol$ of $B$ is  $\Pfun{\Blf} = \PfunER{\pi, y}$.

\end{definition}
	
Notice that when there is consensus about the proposition $p$ of interest, 
i.e., when all agents in belief configuration $\Blf$ hold the same belief 
value, we have $\Pfun{\Blf}{=}0$. 
This happens exactly when all agents' beliefs fall within the same bin of the underlying discretization $D_k$. The following property is an easy consequence from Def.~\ref{def:poler} and Def.~\ref{k-bin:def}. 

\begin{restatable}[Zero Polarization]{proposition}{respolconsensus}
\label{pol-consensus}
Let $D_k{=}I_0,I_1,\ldots, I_{k-1}$ 
be the discretization of $[0,1]$ in Def.~\ref{k-bin:def}. Then 
 $\Pfun{\Blf}{=}0$ iff there exists $m{\in}\{0,\ldots,k{-}1\}$ s.t. for all $i{\in}\Agents,$
 $\Blf_{\agent{i}}{\in}I_m$. 
\end{restatable}

\subsubsection{Dynamic Elements of the Model}
\label{sec:model-dynamic}

\emph{Dynamic elements} formalize the evolution of agents' beliefs as they interact over time and are exposed to different opinions. They include:


\begin{itemize}
\item A \emph{time frame} $\calt{ = }\{0, 1, 2, \ldots \}$ representing the 
discrete passage of time.

\item A \emph{family of belief configurations} $\{\Blft{t}{:}\Agents{\rightarrow}[0,1]\}_{t{\in}{\calt}}$
s.t.\ each $\Blft{t}$ is the belief configuration of agents in 
$\Agents$ w.r.t. proposition $p$ at time step $t{\in}\calt$.

\item A \emph{weighted directed graph} $\Inter{:}\Agents{\times}\Agents {\rightarrow}[0,1].$  The value $\Inter(\agent{i},\agent{j})$, written $\Ifun{\agent{i}}{\agent{j}}$, 
represents the \emph{direct influence} that agent $i$ has on agent $j$, or the \emph{weight} $\agent{i}$ carries with $\agent{j}$. A higher value means stronger weight. Conversely,  $\Inter_{\agent{i},\agent{j}}$ can also be viewed as the \emph{trust} or \emph{confidence} that  $j$ has on $i$.  We assume that $\Ifun{\agent{i}}{\agent{i}}{=}1$, meaning that agents are self-confident. We shall often refer to $\Inter$ simply as the \emph{influence} (graph) $\Inter$. 

We distinguish, however, the direct influence $\Inter_{\agent{i},\agent{j}}$ that $\agent{i}$ has on  $\agent{j}$ 
from the \textit{overall effect} of $\agent{i}$ in $\agent{j}$'s belief.
This effect is a combination of various factors, including
direct influence, their current opinions, the topology of the influence graph, and how agents reason. This overall effect is captured by the update function below.

\item  An \emph{update function} $\Upd{:}(\Blft{t},\Inter){\mapsto}\Blft{t+1}$ 
mapping belief configuration $\Blft{t}$ at time $t$ and influence graph $\Inter$ to new belief configuration $\Blft{t+1}$ at time $t{+}1$.
This function models the evolution of agents' beliefs over time.
We adopt the following premises.
\end{itemize}


\begin{enumerate}[(i)]
    \item \textbf{Agents present some Bayesian reasoning}: 
    Agents' beliefs are updated in every time step by combining their current belief with a \emph{correction term} that incorporates the new evidence they are exposed to
\todo{Santiago to Any: Either "Agents' beliefs are", or "Agents' belief is"}
    in that step --i.e., other agents' opinions.
    More precisely, when agent $\agent{j}$ interacts with agent $\agent{i}$,
    the former affects the latter moving $\agent{i}$'s belief towards $\agent{j}$'s, 
    proportionally to the difference $\Bfun{\agent{j}}{t}{ - }\Bfun{\agent{i}}{t}$ in their beliefs. 
    The intensity of the move is proportional to the influence $\Ifun{\agent{j}}{\agent{i}}$ 
    that $\agent{j}$ carries with $\agent{i}$. 
    The update function produces an overall correction term for each agent as the average of all 
    other agents' effects on that agent,  and then incorporates this term into the agent's current belief.~\footnote{Note that this assumption implies that an agent has an influence on himself, 
    and hence cannot be used as a ``puppet'' who immediately assumes another's agent's belief.} 
    The factor $\Ifun{\agent{j}}{\agent{i}}$  allows the model to capture 
    \emph{authority bias}~\cite{Ramos:19:Book},
    by which agents' influences on each other may have different intensities (by, e.g., giving 
    higher weight to an authority's opinion).
    
    \item \textbf{Agents may be prone to confirmation bias}:
    Agents may give more weight to evidence supporting their 
    current beliefs while discounting evidence contradicting them,
    independently from its source.
    This behavior in known in the psychology literature as 
    \emph{confirmation bias}~\cite{Aronson10}, 
    and is captured in our model as follows.
    When agent $\agent{j}$ interacts with agent $\agent{i}$, the update function moves agent $\agent{i}$'s belief toward 
    that of agent $\agent{j}$, proportionally to the influence $\Ifun{\agent{j}}{\agent{i}}$ of $\agent{j}$ on $\agent{i}$, but with a caveat: the move is stronger when $\agent{j}$’s belief is similar to $\agent{i}$’s than when it is dissimilar. 
\end{enumerate}

The premises above are formally captured in the following update-function.

\begin{definition}[Confirmation-bias] 
\label{def:confirmation-bias}
    Let $\Blft{t}$ be a belief configuration at time $t{\in}\calt$, and  $\Inter$ be an influence graph. The \emph{confirmation-bias update-function} is the map  $\UpdCB{:}({\Blft{t}},{\Inter})\mapsto\Blft{t+1}$ with  $\Blft{t+1}$ given by
    $
        \Bfun{\agent{i}}{t+1} = \Bfun{\agent{i}}{t} + \nicefrac{1}{|\Agents_{\agent{i}}|} \sum_{\agent{j} \in \Agents_{\agent{i}}} 
        \CBfun{i}{j}{t} \, \Ifun{j}{i} \, (\Bfun{j}{t} - \Bfun{i}{t}),
   $
   for every agent $\agent{i}{\in}\Agents$, 
   where $\Agents_{\agent{i}}=\{\agent{j}{\in}\Agents \mid \Inter_{j,i}{>}0 \}$ is the set of \emph{neighbors} of $\agent{i}$ and  $\CBfun{i}{j}{t}{=}1{-}|\Bfun{j}{t}{-}\Bfun{i}{t}|$ is the \emph{confirmation-bias factor} of  $i$ w.r.t. $j$ given their beliefs at time $t$.
    
\end{definition}

The expression $\nicefrac{1}{|\Agents_{\agent{i}}|} \sum_{\agent{j} \in \Agents_{\agent{i}}} \CBfun{i}{j}{t} \, \Ifun{j}{i} \, (\Bfun{j}{t} - \Bfun{i}{t})$ in Def.~\ref{def:confirmation-bias} 
    is a \emph{correction term} incorporated into agent $\agent{i}$'s 
    original belief $\Bfun{\agent{i}}{t}$ at time $t$.
    The correction is the average of the effect of each 
    neighbor $\agent{j}{\in}\Agents_{i}$ on agent $\agent{i}$'s belief 
    at that time step.
    The value  $\Bfun{\agent{i}}{t+1}$ is the resulting updated belief of
    agent $\agent{i}$ at time $t{+}1$.

The confirmation-bias factor $\CBfun{i}{j}{t}$ 
lies in the interval $[0,1]$, and the lower its value, the more agent $\agent{i}$ discounts 
the opinion provided by agent $\agent{j}$ when incorporating it.
It is maximum when agents' beliefs are identical, and minimum they are extreme opposites.

\begin{remark}[Classical Update: Authority Non-Confirmatory Bias]\label{authority:bias:remark}
In this paper we focus on confirmation-bias update and, unless otherwise 
stated, assume the underlying function is given by Def.~\ref{def:confirmation-bias}. 
Nevertheless, in Sections \ref{circulation:section} and \ref{sec:degroot}
we will consider a \emph{classical update} $\UpdR{:}({\Blft{t}},{\Inter}){\mapsto}\Blft{t+1}$ that captures
non-confirmatory authority-bias and
is obtained by replacing the confirmation-bias factor $\CBfun{i}{j}{t}$ 
in Def.~\ref{def:confirmation-bias} with 1.
That is,
$\Bfun{\agent{i}}{t+1}{=}\Bfun{\agent{i}}{t}{+}\nicefrac{1}{|\Agents_{\agent{i}}|} \sum_{\agent{j} \in \Agents_{\agent{i}}} 
         \Ifun{j}{i} \, (\Bfun{j}{t}{-}\Bfun{i}{t}).$
         (We refer to this function as \emph{classical} because it is
         closely related to the standard update function of the DeGroot 
         models for social learning from Economics ~\cite{degroot}. This correspondence will be formalized in Section \ref{sec:degroot}.)
\end{remark}
\subsection{Running Example and Simulations}
\label{sec:simulations}
We now present a running example and several simulations
that motivate our theoretical results. 
Recall that we assume  $\Ifun{\agent{i}}{\agent{i}}{=}1$ for every $i{\in}\Agents$. For simplicity, in all figures of influence graphs we omit self-loops. 

In all cases we compute the polarization measure  (Def.~\ref{k-bin:def}) using 
a discretization $D_k$ of  $[0,1]$ for $k{=}5$ bins, 
each representing a possible general position 
w.r.t. the veracity of the proposition $p$ of interest:
\textit{strongly against}, $[0,0.20)$;
\textit{fairly against}, $[0.20,0.40)$;
\textit{neutral/unsure}, $[0.40,0.60)$;
\textit{fairly in favour}, $[0.60,0.80)$; and
\textit{strongly in favour}, $[0.80,1]$.\footnote{Recall from Def.~\ref{k-bin:def} that our model allows arbitrary discretizations $D_k$ --i.e., different number of bins, with not-necessarily
uniform widths-- depending on the scenario of interest.}
We set parameters $\alpha{=}1.6$, 
as suggested by Esteban and Ray~\cite{Esteban:94:Econometrica},
and $K{=}1\,000$.
In all definitions we let $\Agents{=}\{0, 1,\ldots, n{-}1 \}$, 
and $\agent{i},\agent{j}{\in}\Agents$ be generic agents.

As a running example we consider the following hypothetical situation.

\begin{example}[Vaccine Polarization]
\label{running-example}
Consider the sentence ``vaccines are safe'' as the proposition $p$ of interest. 
Assume a set $\Agents$ of $6$ agents that is initially \emph{extremely polarized} about $p$:
agents 0 and 5 are absolutely confident, respectively, in the falsehood or
truth of $p$, whereas the others are equally split into strongly in favour and strongly against $p$. 

Consider first the situation described by the influence graph in Fig.~\ref{fig:circulation-graph-a}. 
Nodes 0, 1 and 2 represent anti-vaxxers, whereas the rest are pro-vaxxers.  
In particular, note that although initially in total disagreement about $p$, Agent 5 carries a lot of weight with Agent 0. 
In contrast,  Agent $0$'s opinion is very close to that of Agents 1 and 2, even if they do not have any direct influence over him. 
Hence the evolution of Agent $0$'s beliefs will be mostly shaped by that of Agent $5$. 
As can be observed in the evolution of agents' opinions in Fig.~\ref{fig:circulation-graph-b}, Agent 0 moves from being initially strongly against to being fairly in favour of $p$ around time step 8. 
Moreover, polarization eventually vanishes (i.e., becomes zero) around time 20, as agents reach the consensus of being fairly against  $p$.

\begin{figure}[tb]
    \centering
    \begin{subfigure}[t]{.32\textwidth}
      \centering
      \includegraphics[width=1.0\linewidth]{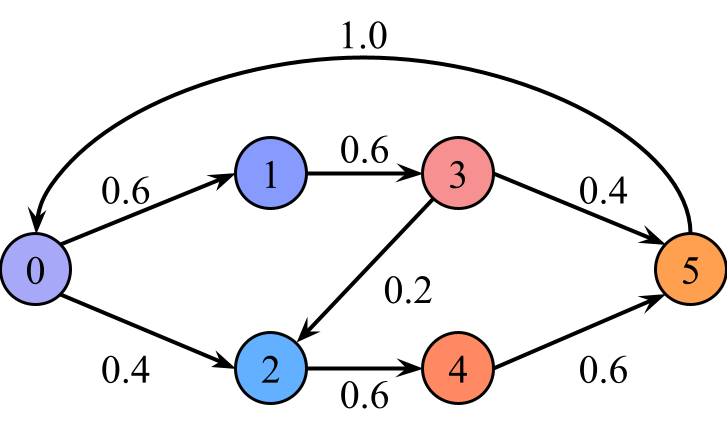}
        \caption{Influence graph $\Inter$ for Ex.~\ref{running-example}.} 
        \label{fig:circulation-graph-a}
    \end{subfigure}
    \hfill
    \begin{subfigure}[t]{.32\textwidth}
      \centering
      \includegraphics[width=1.0\linewidth]{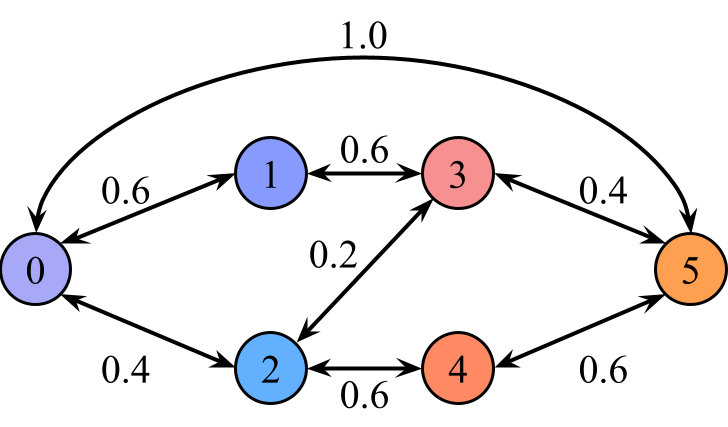}
  \caption{Adding inverse influences to Fig.~\ref{fig:circulation-graph-a}.} 
        \label{fig:circulation-graph-c}
    \end{subfigure}
    \hfill
    \begin{subfigure}[t]{.32\textwidth}
      \centering
      \includegraphics[width=1.0\linewidth]{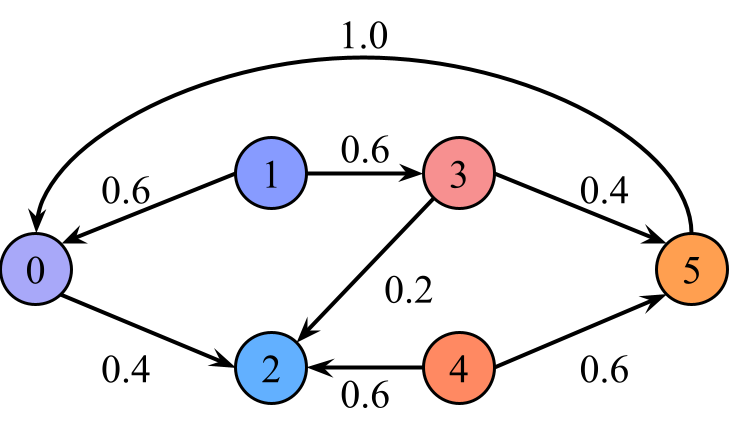}
      \caption{Inversion of $\Ifun{0}{1}$ and $\Ifun{2}{4}$ in Fig.~\ref{fig:circulation-graph-a}.}
      \label{fig:weakly-graph}
    \end{subfigure}
    \\
    \begin{subfigure}[t]{.32\textwidth}
      \centering
      \includegraphics[width=1.0\linewidth]{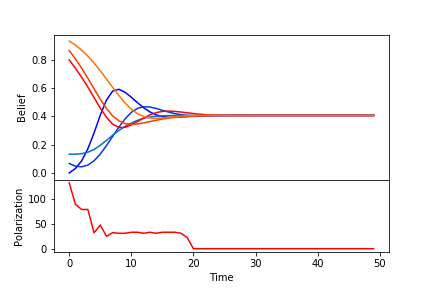}
         \caption{Beliefs and pol. for Fig.~\ref{fig:circulation-graph-a}.} 
        \label{fig:circulation-graph-b}
     \end{subfigure}
     \hfill
    \begin{subfigure}[t]{.32\textwidth}
      \centering
      \includegraphics[width=1.0\linewidth]{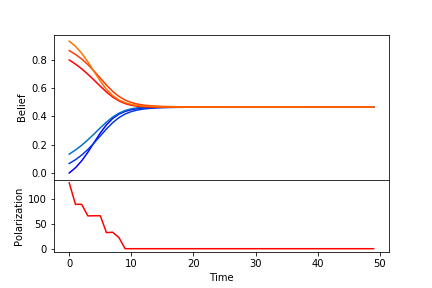}
   \caption{Beliefs and pol. for Fig.~\ref{fig:circulation-graph-c}.} 
        \label{fig:circulation-graph-d}
    \end{subfigure} 
    \hfill
    \begin{subfigure}[t]{.32\textwidth}
      \centering
      \includegraphics[width=1.0\linewidth]{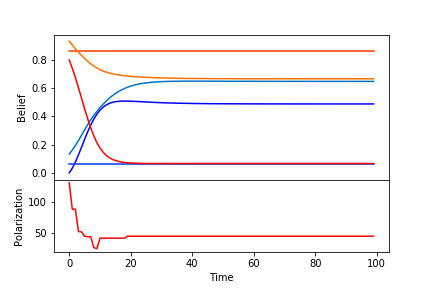}
      \caption{Belief and pol.  for  Fig.~\ref{fig:weakly-graph}.
      }
      \label{fig:weakly-graph-evol}
    \end{subfigure}
    \caption{Influence graphs and evolution of beliefs and polarization for Ex.~\ref{running-example}.}
    \label{running-example:figure}
\end{figure}

Now consider the influence graph in Fig.~\ref{fig:circulation-graph-c}, which is similar to Fig.~\ref{fig:circulation-graph-a}, but with reciprocal influences (i.e., the influence of $i$ over $j$ is the same as the influence of $j$ over $i$). Now  Agents 1 and 2  do have direct influences over Agent 0, so the evolution of Agent $0$'s belief will be partly shaped by initially opposed agents: Agent 5 and the anti-vaxxers. But since Agent $0$'s opinion is very close to that of Agents 1 and 2, the confirmation-bias factor will help keeping 
Agent $0$'s opinion close to their opinion against $p$. In particular, in contrast to the situation in Fig.~\ref{fig:circulation-graph-b},  Agent $0$ never becomes in favour of $p$. The evolution of the agents' opinions and their polarization is shown in Fig.~\ref{fig:circulation-graph-d}. Notice that polarization vanishes around time 8 as the agents reach consensus but this time they are more positive about (less against) $p$ than in the first situation. 

Finally, consider the situation in Fig.~\ref{fig:weakly-graph} obtained from Fig.~\ref{fig:circulation-graph-a} by inverting the influences of Agent 0 over Agent 1 and 
Agent 2 over Agent 4. Notice that Agents 1 and 4 are no longer influenced by anyone though they influence others. Thus, as shown in Fig.\ref{fig:weakly-graph-evol}, their beliefs do not change over time,
which means that the group does not reach consensus and polarization never disappears though it is considerably reduced. \qed 

\end{example}
The above example illustrates complex non-monotonic,  overlapping, convergent, and non-convergent evolution of agent beliefs and polarization 
even in a small case with $n{=}6$ agents. 
Next we  present simulations for several influence graph topologies 
with $n{=}1\,000$ agents, which illustrate more of this complex behavior
emerging from con\-fir\-ma\-tion-bias interaction among agents.
Our theoretical results in the next sections bring insight into 
the evolution of beliefs and polarization depending on graph topologies. 

In all simulations we limit execution to $\tmax$ time 
steps varying according to the experiment.
\version{A detailed mathematical specification of simulations can be found in the corresponding technical report~\cite{Alvim:21:ForteTechRep}.}{The complete mathematical specification of simulations are given
in Appendix~\ref{sec:simulations-specification}.}

We consider the following initial belief configurations, depicted in Fig.~\ref{fig:initial-belief-configurations}:
    a \emph{uniform} belief configuration with a set of agents whose beliefs are as varied as possible, all equally spaced in the interval $[0, 1]$;
    a \emph{mildly polarized} belief configuration with agents evenly split into two groups with moderately dissimilar inter-group beliefs compared to intra-group beliefs;
    an \emph{extremely polarized} belief configuration
    representing a situation in which half of the 
    agents strongly believe the proposition, whereas 
    half strongly disbelieve it; and
    a \emph{tripolar} configuration with 
    agents divided into three 
    groups.
    
\begin{figure}[tb]
    \centering
    \includegraphics[width=\linewidth]{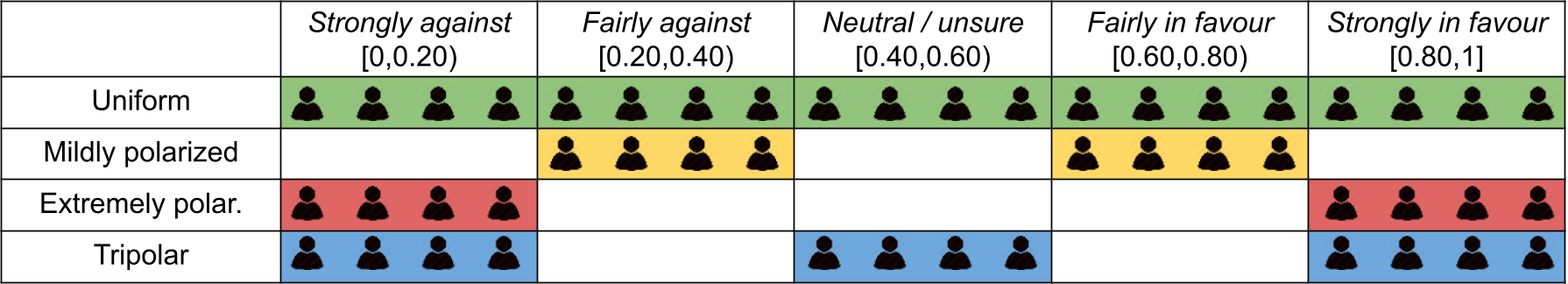}
    \caption{Depiction of different 
    initial belief configurations used in simulations.
    }
    \label{fig:initial-belief-configurations}
\end{figure}

As for influence graphs, we consider the following ones, depicted in Fig.~\ref{fig:interaction-graphs}:

\begin{figure}[tb]
    \centering
    \begin{subfigure}[t]{0.15\textwidth}
      \centering
      \includegraphics[width=\textwidth]{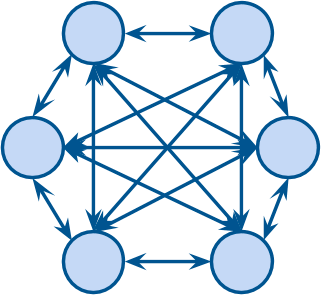}
      \caption{Clique}
      \label{fig:interaction-graphs-clique}
    \end{subfigure}
    \hfill
    \begin{subfigure}[t]{0.15\textwidth}
      \centering
      \includegraphics[width=\textwidth]{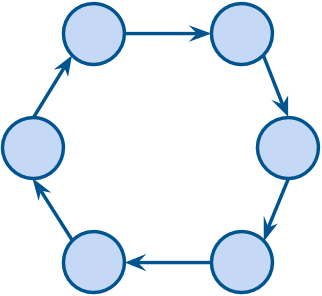}
      \caption{Circular}
      \label{fig:interaction-graphs-circular}
    \end{subfigure}
    \hfill
    \begin{subfigure}[t]{0.30\textwidth}
      \centering
      \includegraphics[width=0.9\textwidth]{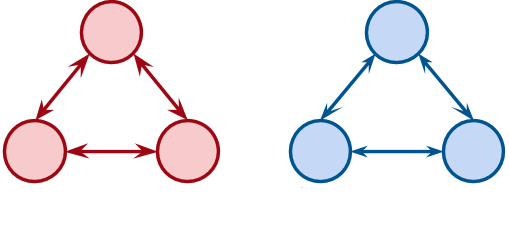}
      \caption{Disconnected groups}
      \label{fig:interaction-graphs-disconnected}
    \end{subfigure}
    \hfill
    \begin{subfigure}[t]{0.33\textwidth}
      \centering
      \includegraphics[width=0.8\textwidth]{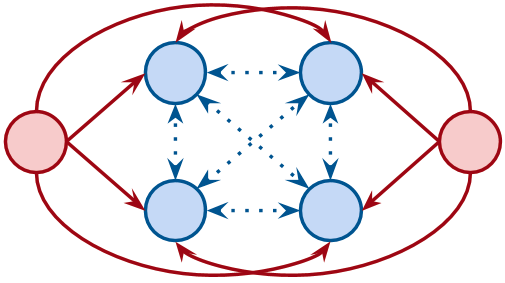}
      \caption{Unrelenting influencers}
      \label{fig:interaction-graphs-unrelenting}
    \end{subfigure}
    \caption{The general shape of influence graphs used in simulations, for $n{=}6$ agents.}
    \label{fig:interaction-graphs}
\end{figure}

\begin{itemize}
    \item A \emph{$C$-clique} influence graph $\Interclique$ in which each agent influences every other with constant 
    value $C{=}0.5$.
    This represents a social network in
    which all agents interact among themselves, and are all immune
    to authority bias.
    
    \item A \emph{circular} influence graph $\Intercircular$ 
    representing a social network in which agents can be organized
    in a circle in such a way each agent is only influenced by its 
    predecessor and only influences its successor.
    This is a simple instance of a balanced graph (in which
    each agent's influence on others is as high as the influence received, 
    as in Def.~\ref{def:circulation} ahead), which is a pattern commonly 
    encountered in some sub-networks.
    
    \item A \emph{disconnected} influence graph $\Interdisconnected$
    representing a social network sharply divided into two groups in such a way that agents within the same group 
    can considerably influence each other, but not at all
    the agents in the other group.

    
    \item An \emph{unrelenting influencers} influence graph $\Interunrelenting$ 
    representing a scenario in which two agents 
    exert significantly stronger influence on every other agent than these other agents have among themselves.
    This could represent, e.g., a social network 
    in which two totalitarian media companies dominate the news
    market, both with similarly high levels of influence on 
    all agents.
    The networks have clear agendas to push forward, and are not influenced in a meaningful way by other agents.
\end{itemize}


We simulated the evolution of agents' beliefs
and the corresponding polarization of the network 
for all combinations of initial belief configurations 
and influence graphs presented above.
The results, depicted in Figure~\ref{fig:comparing-num-bins},
will be used throughout this paper to illustrate some of our 
formal results. 
Both the Python implementation of the model and the Jupyter Notebook 
containing the simulations are available on Github~\cite{website:github-repo}.



\begin{figure}[tbp]
\centering
  \includegraphics[width=1\linewidth]{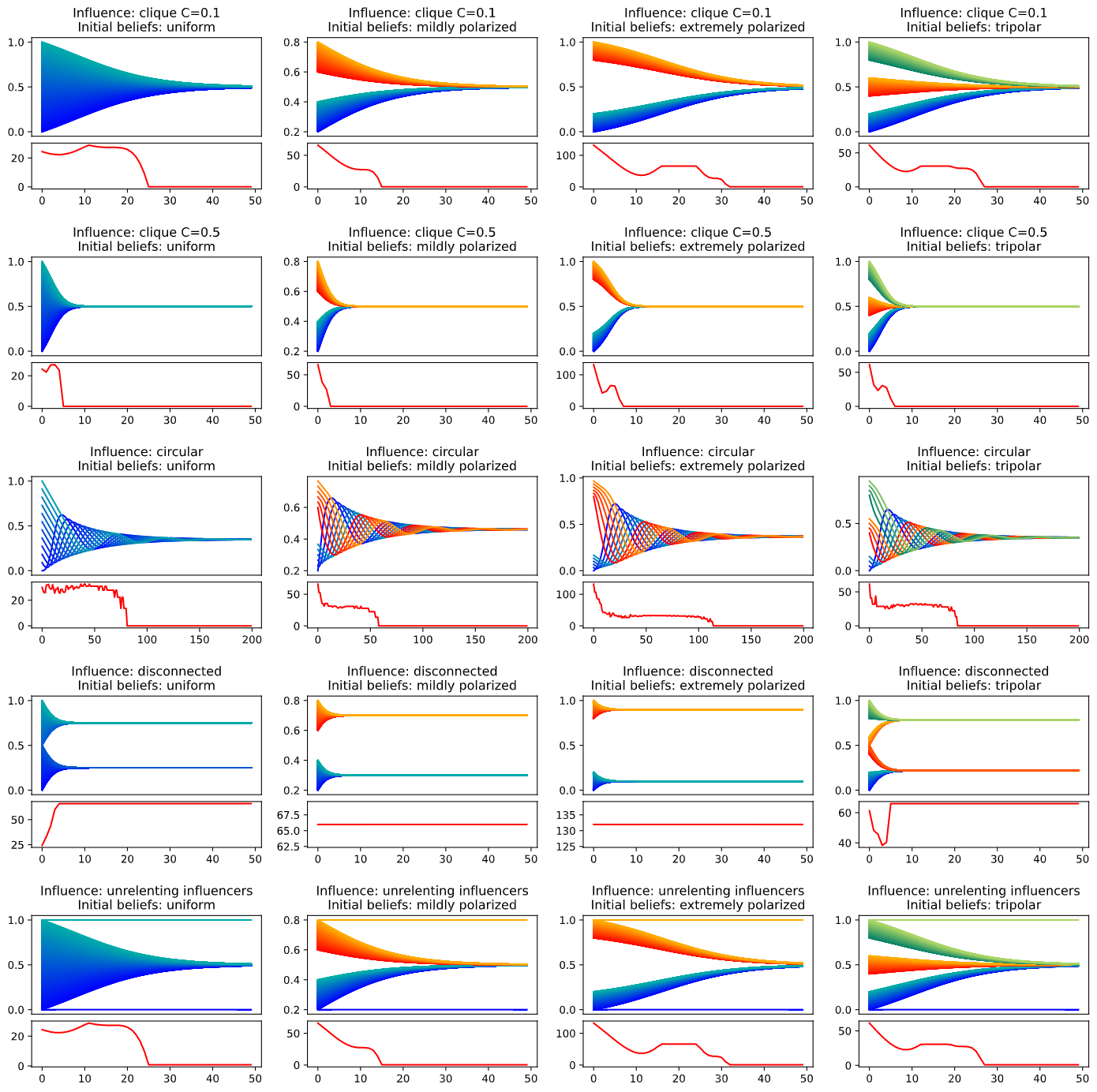}
  \caption{Evolution of belief  and  polarization 
  under confirmation bias. Horizontal axes represent time.
  Each row contains all graphs with the same influence graph, and each column 
  all graphs with the same initial belief configuration.
  Simulations of circular influences used $n{=}12$ agents, the rest used $n{=}1\,000$ agents.
  }
  \label{fig:comparing-num-bins}
\end{figure}
\section{Belief and Polarization Convergence}
\label{sec:general-result}

Polarization tends to diminish as agents approximate a \emph{consensus}, i.e.,
as they (asymptotically) agree upon a common belief value for the proposition of interest.  Here and in Section \ref{sec:specific-cases} we consider meaningful families of influence graphs that guarantee consensus \emph{under confirmation bias}.
We also identify fundamental properties of agents,
and the value of convergence. 
Importantly, we relate influence with the notion of \emph{flow} in flow networks, and use it to identify necessary conditions for polarization not converging to zero.

 
 \subsection{Polarization at the limit}
 
Prop.~\ref{pol-consensus} states that our polarization measure on 
a belief configuration  (Def.~\ref{k-bin:def}) is zero exactly when all 
belief values in it lie within the same bin of the underlying discretization
$D_k{=}I_0\ldots I_{k-1}$ of $[0,1]$. 
In our model polarization converges to zero if all agents' beliefs converge 
to a same non-borderline value. More precisely:

\begin{restatable}[Zero Limit Polarization]{lemma}{respolatlimit}
\label{pol-at-limit}
 Let $v$ be a non-borderline point of $D_k$ such that for every $\agent{i}{\in} \Agents$, $\lim_{t \to \infty} \Bfun{i}{t}{=}v.$ Then $\lim_{t \to \infty} \Pfun{\Blft{t}}{=}0$. 
\end{restatable}
 
To see why we exclude the $k{-}1$
borderline values of $D_k$ in the above lemma, assume $v{\in}I_m$ is a borderline value. Suppose that there are two agents $i$ and $j$ whose beliefs converge to $v$, 
but with the belief of $i$ staying always within $I_m$ whereas the belief of $j$ remains outside of $I_m$.
Under these conditions one can verify, using Def.~\ref{def:poler} and Def.~\ref{k-bin:def}, that $\Pol$ will not converge to $0$. This situation is illustrated in Fig.~\ref{fig:borderline-2bin-polarization} assuming a discretization $D_2 =[0, \nicefrac{1}{2}), [\nicefrac{1}{2},1]$ whose only borderline is $\nicefrac{1}{2}$. Agents' beliefs converge to value $v{=}\nicefrac{1}{2}$, but polarization does not converge to 0. In contrast,  Fig.\ref{fig:borderline-3bin-polarization} illustrates Lem.\ref{pol-at-limit} for  $D_3=[0,\nicefrac{1}{3}), [\nicefrac{1}{3},\nicefrac{2}{3}), [\nicefrac{2}{3},1].$
~\footnote{It is worthwhile to note that this discontinuity at borderline points matches real scenarios where each bin represents a sharp action an agent takes based on his current belief value.
Even when two agents' beliefs are asymptotically converging to a same borderline value from different sides, their discrete decisions will remain distinct. 
E.g., in the vaccine case of Ex.~\ref{running-example}, even agents that are
asymptotically converging to a common belief value of $0.5$ will take different
decisions on whether or not to vaccinate, depending on which side of $0.5$ their belief falls.
In this sense, although there is convergence in the underlying belief values, there remains polarization w.r.t. real-world actions taken by agents.}

 \begin{figure}[tb]
    \centering
    \begin{subfigure}[t]{.32\textwidth}
      \centering
      \includegraphics[width=\linewidth]{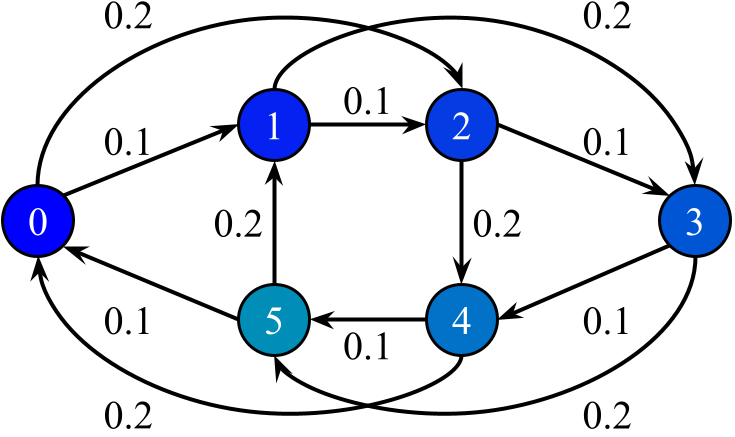}
      \caption{Influence graph.}
      \label{fig:double-circular-graph}
    \end{subfigure}
    \hfill
    \begin{subfigure}[t]{0.32\textwidth}
      \centering
      \includegraphics[width=0.9\linewidth]{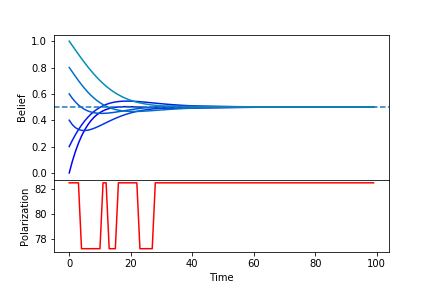}
      
      \caption{Beliefs and polarization, 2 bins, for graph in Fig.~\ref{fig:double-circular-graph}.  
      }\label{fig:borderline-2bin-polarization}
    \end{subfigure}
    \hfill
    \begin{subfigure}[t]{0.32\textwidth}
      \centering
      \includegraphics[width=0.9\linewidth]{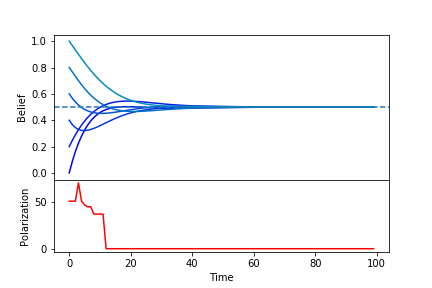}
      \caption{Beliefs and polarization, 3 bins, for graph in Fig.~\ref{fig:double-circular-graph}.  
      }\label{fig:borderline-3bin-polarization}
    \end{subfigure}
    \caption{Belief convergence to borderline value 1/2. Polarization does not converge to 0 with equal-length 2 bins (Fig.~\ref{fig:borderline-2bin-polarization}) and but it does with 3 equal-length  bins (Fig.~\ref{fig:borderline-3bin-polarization}).}
    \label{fig:borderline-example}
\end{figure}


\subsection{Convergence under Confirmation Bias in Strongly Connected Influence}
\label{sec:main-result}


We now introduce the family of \emph{strongly-connected} influence graphs, which includes cliques, that describes scenarios where each agent has an influence over all others. Such influence is not necessarily \emph{direct} in the sense defined next, or the same for all agents, as in the more specific cases of cliques.

\begin{definition}[Influence Paths] \label{def:influence-path} Let $C \in (0,1].$
We say that $\agent{i}$ has a \emph{direct influence} $C$ over $\agent{j}$, written  $\ldinfl{\agent{i}}{C}{\agent{j}}$, if $\Ifun{\agent{i}}{\agent{j}}=C.$ 
 
 An \emph{influence path} is a {finite sequence} of \emph{distinct} agents from $\Agents$ where each agent in the sequence has a direct influence over the next one. Let  $p$ be an influence path $i_0i_1\ldots i_n.$ The \emph{size} of $p$ is $|p|{=}n$. 
We also use $\ldinfl{\agent{i_0}}{C_1}{\agent{i_1}}\ldinfl{}{C_2}{}\ldots \ldinfl{}{C_n} {\agent{i_n}}$ to denote  $p$ with the direct influences along this path. We write  $\linfl{\agent{i_0}}{C}{p}{\agent{i_n}}$ to indicate that the \emph{product influence} of $i_0$ over $i_n$ along  $p$ is $C{=}C_1{\times}\ldots {\times}C_n$. 

 We often omit influence or path indices from the above arrow notations when they are unimportant or clear from the context.  We say that   $\agent{i}$ has an \emph{influence} over $\agent{j}$ if $\infl{\agent{i}}{\agent{j}}$. 
\end{definition}

The next definition is akin to the graph-theoretical notion of strong connectivity. 
\begin{definition}[Strongly Connected Influence]  We say that an influence graph $\Inter$ is \emph{strongly connected} if for all $\agent{i}$, $\agent{j}{\in} \Agents$ such that $i{\neq}j$, $\infl{\agent{i}}{\agent{j}}$.
\end{definition}

\begin{remark}\label{rmk:assumption}  For technical reasons we assume that, \emph{initially}, there are no two agents $\agent{i}, \agent{j}{\in}\Agents$ such that $\Bfun{i}{0}{=} 0$ and $\Bfun{j}{0}{=}1.$ This implies that 
 for every $\agent{i}, \agent{j}{\in}\Agents$: $\CBfun{\agent{i}}{\agent{j}}{0}{>} 0$ where  $\CBfun{\agent{i}}{\agent{j}}{0}$ is the confirmation bias of $i$ towards $j$ at time $0$ (See Def.~\ref{def:confirmation-bias}). Nevertheless, at the end of this section we will address the cases in which this condition does not hold.
\end{remark}

We shall use the notion of maximum and minimum belief values at a given time $t$. 

\begin{definition}[Extreme Beliefs]\label{def:extreme:beliefs} Define $\mx{t}=\max_{\agent{i}{\in}\Agents} \Bfun{\agent{i}}{t}$ and $\mn{t}= \max_{\agent{i}{\in}\Agents} \Bfun{\agent{i}}{t}.$
 \end{definition}

It is worth noticing that \emph{extreme agents} --i.e., those holding extreme beliefs-- do not necessarily remain the same across time steps. 
Fig.~\ref{fig:circulation-graph-b} illustrates this point: 
Agent 0 goes from being the one most against the proposition of interest
at time $t{=}0$ to being the one most in favour of it around $t{=}8$.
Also, the third row of Fig.~\ref{fig:comparing-num-bins} shows simulations
for a circular graph under several initial belief configurations. 
Note that under all initial belief configurations
different agents alternate as maximal and minimal belief holders.

Nevertheless, in what follows will show that the beliefs of all agents, under strongly-connected influence  and confirmation bias, converge to the same value since the difference between $\mn{t}$ and $\mx{t}$ goes to 0 as $t$ approaches infinity. 
\todo{Bernardo: before this part said that it was a distinctive property of strongly connected graphs, but this is valid for all graphs}
We begin with a  lemma stating a property of the confirmation-bias update: \emph{The belief value of any agent at any time is bounded by those from extreme agents in the previous time unit}. 

\begin{restatable}[Belief Extremal Bounds]{lemma}{reslemmacbmaxdiffmin}
\label{lemma:cb-maxdiffmin}
 For every $i\in \Agents$,   $\mn{t} \leq \Bfun{\agent{i}}{t{+}1} \leq \mx{t}.$
\end{restatable}

The next corollary follows from the assumption in Rmk.~\ref{rmk:assumption} and Lemma \ref{lemma:cb-maxdiffmin}.

\begin{restatable}[]{corollary}{rescorbiasfactor}
\label{cor:biasfactor} 
For every $\agent{i}, \agent{j}{\in}\Agents$, $t{\geq}0$: $\CBfun{\agent{i}}{\agent{j}}{t}{>}0$.
\end{restatable}


Note that monotonicity does not necessarily hold for belief evolution. This 
is illustrated by Agent 0's behavior in Fig.~\ref{fig:circulation-graph-b}. However, it follows immediately from Lemma~\ref{lemma:cb-maxdiffmin} that $\mn{\cdot}$ and $\mx{\cdot}$ are monotonically increasing and decreasing functions of $t$. 
    
\begin{restatable}[Monotonicity of Extreme Beliefs]{corollary}{rescorcbmbeforemafter}   \label{cor:cb-mbefore-mafter} 
$\mx{t+1}{\leq}\mx{t}$ and $min^{t+1}{\geq}\mn{t}$ for all $t{\in}\nat$.
 \end{restatable}

Monotonicity and the bounding of $\mx{\cdot}$, $\mn{\cdot}$ within $[0,1]$ lead us, via the Monotonic Convergence Theorem \cite{Sohrab:14}, to the existence of \emph{limits for beliefs of extreme agents}.

\begin{restatable}[Limits of Extreme Beliefs]{theorem}{rescorcbmaxlimitsexist}
\label{th:cb-max-limits-exist} There are $U,L{\in}[0,1]$ s.t.
$\lim_{t\to\infty} \mx{t}{=}U$ and $\lim_{t\to\infty} \mn{t}{=}L.$ 
\end{restatable}

We still need to show that $U$ and $L$ are the same value. For this we prove a distinctive property of agents under strongly connected  influence graphs: the belief of any agent at time $t$ will influence every other agent by the time $t{+}|\Agents|{-}1$. This is precisely formalized below in Lemma \ref{lemma:cb-path-bound}. 
First, however, we introduce some bounds for confirmation-bias, influence as well as notation for the limits in Th.\ref{th:cb-max-limits-exist}.

\begin{definition}[Min Factors]\label{def:fcb-min}
Define $\CBfunM{=}\min_{\agent{i}, \agent{j} \in \Agents} \CBfun{i}{j}{0}$ as the minimal confirmation bias factor at $t{=}0$. Also let $\IfunM$ be the smallest positive influence in $\Inter$.  Furthermore, let 
$L{=}\lim_{t\to\infty} \mn{t}$ and $U{=}\lim_{t\to\infty} \mx{t}.$ 
\end{definition}

Notice that since $\mn{t}$ and $\mx{t}$ do not get further apart as the time $t$ increases (Cor.~\ref{cor:cb-mbefore-mafter}),  $\min_{\agent{i}, \agent{j}{\in} \Agents} \CBfun{i}{j}{t}$  is a non-decreasing function of $t$. Therefore $\CBfunM$  acts as a lower bound for the confirmation-bias factor in every time step.
\begin{restatable}[]{proposition}{resfcbminprop}
\label{fcb-min:prop}
$\CBfunM =  \min_{\agent{i}, \agent{j} \in \Agents} \CBfun{i}{j}{t} $ for every $t > 0$. 
\end{restatable}

The factor $\CBfunM$ is used in the next result to establish that the belief of agent $\agent{i}$ at time $t$, the minimum confirmation-bias factor, and the maximum belief at $t$ act as bound of the belief of $\agent{j}$ at $t{+}|p|$, 
where $p$ is an influence path from $\agent{i}$ and $\agent{j}$.

\begin{restatable}[Path bound]{lemma}{reslemmacbpathbound}
\label{lemma:cb-path-bound} 
If $\Inter$ is strongly connected:
\begin{enumerate}
    \item Let  $p$ be an arbitrary path $\linfl{\agent{i}}{C}{p}{\agent{j}}$.  Then $\Bfun{\agent{j}}{t+|p|} \leq \mx{t} + \nicefrac{C\CBfunM^{|p|}}{|\Agents|^{|p|}}(
\Bfun{\agent{i}}{t} - \mx{t}).$

    \item  Let $\mstar^t{\in}\Agents$ be an agent holding the least belief value at time $t$ and $p$ be a path such that $\linfl{\agent{\mstar}^t}{}{p}{\agent{i}}$. Then $\Bfun{\agent{i}}{t{+}|p|} \leq \mx{t}{-}\delta$, with $\delta = \left(\nicefrac{\IfunM\CBfunM}{|\Agents|}\right)^{|p|}(U{-}L)$.
\end{enumerate}
\end{restatable}

Next we establish that all beliefs at time $t{+}|\Agents|{-}1$ are smaller
than the maximal belief at $t$ by a factor of at least $\epsilon$ depending
on the minimal confirmation bias, minimal influence and the limit values $L$ 
and $U$.

\begin{restatable}[]{lemma}{reslemcbepsilonbound}
\label{lem:cb-epsilon-bound}
Suppose that $\Inter$ is strongly-connected.
\begin{enumerate}
    \item If $\Bfun{\agent{i}}{t{+}n} \leq \mx{t} - \gamma$ and $\gamma \geq 0$ then $\Bfun{\agent{i}}{t+n+1} \leq \mx{t} - \nicefrac{\gamma}{|\Agents|}.$
    
    \item $\Bfun{\agent{i}}{t+|\Agents|-1} \leq \mx{t} - \epsilon$, where $\epsilon$ is equal to $\left(\nicefrac{\IfunM\CBfunM}{|\Agents|}\right)^{|\Agents|-1}(U-L)$.
\end{enumerate}
\end{restatable}

Lem.~\ref{lem:cb-epsilon-bound}(2)
states that $\max^{\cdot}$ decreases by at least $\epsilon$ after  $|A|{-}1$ steps. Therefore, after $m(|A|-1)$ steps it should decrease by at least $m\epsilon$. 

\begin{restatable}[]{corollary}{rescormaxdiff}
\label{cor:max-dif} 
If $\Inter$ is strongly connected,
$\mx{t+m(|\Agents|-1)}{\leq}\mx{t}{-}m\epsilon$ for $\epsilon$ in  Lem.~\ref{lem:cb-epsilon-bound}.
\end{restatable}

We can now state that in strongly connected influence graphs extreme beliefs eventually converge to the same value. The proof uses Cor.~\ref{cor:biasfactor} and Cor.~\ref{cor:max-dif} above.

\begin{restatable}[]{theorem}{resthul}
\label{th:U=L} 
If $\Inter$ is strongly connected  then 
$\lim_{t\to\infty} \mx{t} =  \lim_{t\to\infty} \mn{t}.$ 
\end{restatable}

Combining Th.~\ref{th:U=L}, the assumption in  Rmk.~\ref{rmk:assumption} and the Squeeze Theorem, we conclude that for strongly-connected graphs, all agents' beliefs converge to the same value.

\begin{restatable}[]{corollary}{restheoremcbsccconvergence}
\label{cor:cb-scc-convergence} 
If $\Inter$ is strongly connected then for all
$\agent{i},\agent{j}{\in}\Agents, \lim_{t \to \infty} \Bfun{\agent{i}}{t} {=}\lim_{t \to \infty} \Bfun{\agent{j}}{t}.$
\end{restatable}

\subsubsection{The Extreme Cases.}\label{extreme:case} We assumed in Rmk.~\ref{rmk:assumption} that there were no two agents $\agent{i},\agent{j}$ s.t. $\Bfun{\agent{i}}{t}{=}0$ and $\Bfun{j}{t}{=}1$. 
Th.~\ref{theorem:cb-geberal-scc-convergence} below addresses
the situation in which this does not happen.
More precisely, it establishes that under confirmation-bias update, 
in any strongly-connected, non-radical society, agents' 
beliefs eventually converge to the same value.

\begin{definition}[Radical Beliefs]\label{radical:def} 
An agent $\agent{i}{\in}\Agents$ is called \emph{radical} if $\Blf_{\agent{i}}{=}0$ or $\Blf_{\agent{i}}{=}1$.
A belief configuration $\Blf$ is \emph{radical} if every $\agent{i}{\in}\Agents$ is radical.
\end{definition}




\begin{restatable}[Confirmation-Bias Belief Convergence]{theorem}{restheoremcbgeberalsccconvergence}
\label{theorem:cb-geberal-scc-convergence}
In a strongly connected influence graph and under the confirmation-bias update-function, if $\Blft{0}$ is not radical then for all $\agent{i}, \agent{j} {\in}\Agents$, $\lim_{t \to \infty} \Bfun{i}{t}{=}\lim_{t \to \infty} \Bfun{j}{t}$. Otherwise for every $\agent{i}{\in}\Agents$, $\Bfun{i}{t}{=}\Bfun{i}{t+1}{\in}\{0,1\}$.
\end{restatable}

We conclude this section by emphasizing that belief convergence is not guaranteed
in non strongly-connected graphs.
Fig.~\ref{fig:weakly-graph} from the vaccine example shows such
a graph where neither belief convergence nor zero-polarization is 
obtained. 

\section{Conditions for Polarization}
\label{sec:specific-cases}

We now use concepts from flow networks to  identify insightful
necessary conditions for polarization never disappearing.  Understanding the conditions when polarization
\emph{does not} disappear under confirmation bias is one of the main contributions of this paper.

\subsubsection{Balanced Influence: Circulations}

The following notion is inspired by the \emph{circulation problem} for directed graphs (or flow network) \cite{Diestel:17}.  
Given a graph $G=(V,E)$ and a function $c{:}E{\to}\reals$ (called \emph{capacity}), the problem involves finding a function  $f{:}E{\to}\reals$ (called \emph{flow}) such that: (1) $f(e){\leq}c(e)$ for each $e{\in} E$; and (2) $\sum_{(v,w){\in} E}f(v,w)=\sum_{(w,v){\in}E}f(w,v)$ for all $v{\in}V$. If such an $f$ exists it  is called a \emph{circulation} for $G$ and $c$. 

Thinking of flow as influence, the second condition, called \emph{flow conservation}, corresponds to requiring that each agent influences others as much as is influenced by them.

\begin{definition}[Balanced Influence] \label{def:circulation} We say that  $\Inter$ is \emph{balanced} (or a \emph{circulation}) if every $\agent{i} \in \Agents$ satisfies the constraint $\sum_{\agent{j}{\in}\Agents} \Ifun{i}{j}{=}\sum_{\agent{j}{\in}\Agents} \Ifun{j}{i}.$
\end{definition}

Cliques and circular graphs, where all (non-self) influence values are equal, are balanced (see Fig.~\ref{fig:interaction-graphs-circular}). 
The graph of our vaccine example (Fig.~\ref{running-example:figure}) is a circulation that it is neither a clique nor a circular graph.
Clearly, influence graph $\Inter$ is balanced if it is a solution to a circulation problem for some $G{=}(\Agents,\Agents{\times}\Agents)$ with capacity $c{:}\Agents{\times}\Agents{\to}[0,1].$  

Next we use a fundamental property from flow networks describing flow conservation for graph cuts \cite{Diestel:17}. Interpreted in our case it says that any group of agents $A{\subseteq}\Agents$ influences other groups as much as they influence $A$. 

\begin{restatable}[Group Influence Conservation]{proposition}{respropgroupinfluenceconservation}
\label{prop:group-influence-conservation}  
Let $\Inter$ be balanced and $\{A,B\}$ be a partition of $\Agents$. Then 
$\sum_{i \in A}\sum_{j \in B} \Ifun{i}{j} =\sum_{i \in A}\sum_{j \in B} \Ifun{j}{i}$. 
\end{restatable}

We now define \emph{weakly connected influence}. 
Recall that an undirected graph is \emph{connected} if there is path between each pair of nodes. 

\begin{definition}[Weakly Connected Influence]\label{def:weakly-connected}
Given an influence graph $\Inter$, define the undirected graph $G_{\Inter}{=}(\Agents,E)$ where $\{i,j \}{\in}E$ if and only if $\Ifun{i}{j}{>}0$ or $\Ifun{j}{i}{>}0$.
An influence graph $\Inter$ is called \emph{weakly connected} if  the undirected graph $G_{\Inter}$ is connected.
\end{definition}

\todo{Mário: I added a ``clearpage'' command here, which is not desirable, but it forces latex to do the right thing and save
us space for the conclusion. I don't know why, but this magical 
``clearpage'' makes it all work for now! We can (and should!) get rid of it in the camera-ready version.}
Weakly connected influence relaxes its strongly connected counterpart. 
However, every balanced, weakly connected influence is strongly connected as implied by the next lemma. Intuitively, circulation flows never leaves strongly connected components. 

\begin{restatable}[]{lemma}{resprocirculationpath}
\label{prop:circulation-path}
If $\Inter$ is balanced and $\Ifun{i}{j}{>}0$ then $\Path{j}{i}$.
\end{restatable}

\subsubsection{Conditions for Polarization}
We have now all elements to identify conditions for permanent polarization. The  convergence for strongly connected graphs (Th.~\ref{theorem:cb-geberal-scc-convergence}), the polarization at the limit lemma (Lem.~\ref{pol-at-limit}), and Lem.~\ref{prop:circulation-path} yield the following noteworthy result.

\begin{restatable}[Conditions for Polarization]{theorem}{respolnonzero}
\label{pol-non-zero}
Suppose that $\lim_{t \to \infty}\Pfun{\Blft{t}}{\neq}0.$ Then either: 
(1) $\Inter$ is not balanced;
(2) $\Inter$ is not weakly connected;
(3) $\Blft{0}$ is radical; or
(4) for some borderline value $v$, $\lim_{t \to \infty} \Bfun{i}{t}{=}v$ for each $\agent{i}{\in}\Agents$ . 
\end{restatable}

%
%

Hence, at least one of the four conditions is necessary for the persistence of polarization. 
If (1) then there must be at least one agent that influences more than what he is influenced (or vice versa). This is illustrated in Fig.~\ref{fig:weakly-graph} from the vaccine example, where Agent 2 is such an agent. 
If (2) then there must be isolated subgroups of agents; e.g., two isolated strongly-connected components the members of the same component will achieve consensus but the consensus values of the two components may be very different. This is illustrated in the fourth row of Fig.~\ref{fig:comparing-num-bins}. Condition (3) can be ruled out if there is an agent that is not radical, like in all of our examples and simulations. As already discussed, (4) depends on the underlying discretization $D_k$ (e.g., assuming equal-length bins if $v$ is  borderline in $D_k$ it is not borderline in $D_{k+1}$, see Fig.~\ref{fig:borderline-example}.).

\subsubsection{Reciprocal and Regular Circulations}
\label{circulation:section}
The notion of circulation allowed us to identify potential causes of polarization. In this section we will also use it to identify meaningful topologies whose symmetry can help us predict the exact belief value of convergence. 

A \emph{reciprocal} influence graph is a 
circulation where the influence of $i$ over $j$ is the
same as that of $j$ over $i$, i.e, $\Ifun{i}{j}{=} \Ifun{j}{i}$.  Also a graph is (\emph{in-degree}) 
\emph{regular} if the in-degree 
%
of each nodes is the same;  i.e., for all $i,j{\in}\Agents$,
$|\Agents_{\agent{i}}|{=}|\Agents_{\agent{j}}|$.  

As examples of regular and reciprocal graphs, consider a graph $\Inter$  where all (non-self) influence values are equal. If $\Inter$ is \emph{circular} then it is a regular circulation, and if $\Inter$ 
is a \emph{clique} then it is a reciprocal regular circulation. Also we can modify slightly our vaccine example to obtain a regular reciprocal circulation as shown in Fig.~\ref{fig:circulation-reg-rec}. 

\begin{figure}[tb]
  \centering
    \begin{subfigure}{.45\textwidth}
      \centering
      \includegraphics[width=0.7\linewidth]{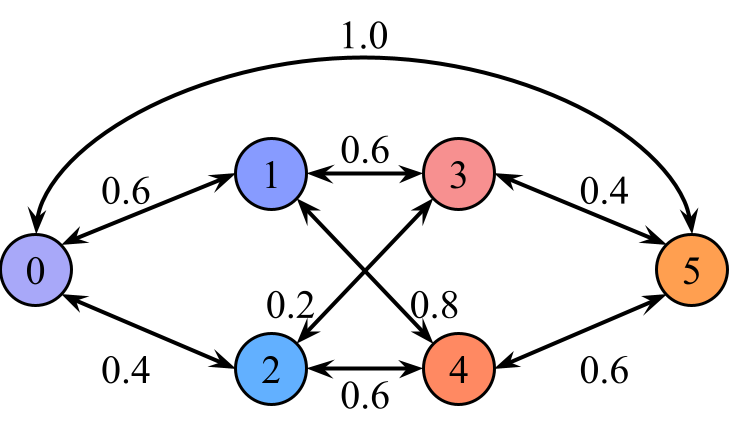}
    \caption{Regular and reciprocal influence.} 
        \label{fig:circulation-reg-rec-a}
    \end{subfigure}
    \qquad
    \begin{subfigure}{.45\textwidth}
      \centering
      \includegraphics[width=0.7\linewidth]{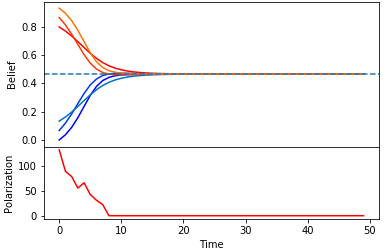}
        \caption{Beliefs and pol. for Fig.~\ref{fig:circulation-reg-rec-a}. 
        }    
        \label{fig:circulation-reg-rec-b}
    \end{subfigure}
    \caption{Influence and evolution of beliefs and polar.}
    \label{fig:circulation-reg-rec}
\end{figure}

The importance of regularity and reciprocity of influence graphs is that their symmetry is sufficient to the determine the exact value all the agents converge to under 
confirmation bias: \emph{the average of initial beliefs}.  Furthermore, under classical update (see Rmk.~\ref{authority:bias:remark}), we can drop reciprocity and obtain the same result. The result is proven using Lem.~\ref{prop:circulation-path}, Th.~\ref{theorem:cb-geberal-scc-convergence}, Cor.~\ref{cor:degroot},  the squeeze theorem and by showing that
$\sum_{\agent{i}{\in}\Agents} \Bfun{i}{t}{=}\sum_{\agent{i}{\in} \Agents} \Bfun{i}{t+1}$ using symmetries derived from reciprocity, regularity, and the fact that $\CBfun{i}{j}{t}{=}\CBfun{j}{i}{t}.$ 

\begin{restatable}[Consensus Value]{theorem}{rescorcirculationconvergence}
\label{cor:circulation-convergence} 
Suppose that $\Inter$ is regular and weakly connected. If $\Inter$ is reciprocal and the belief update is confirmation-bias, or if the influence graph $\Inter$ is a circulation and the belief update is classical, then  $\lim_{t \to \infty} \Bfun{i}{t} = \nicefrac{1}{|\Agents|}\sum_{\agent{j} \in \Agents}\Bfun{j}{0}$ for every $i{\in}\Agents.$
\end{restatable}

\section{Comparison to DeGroot's model}
\label{sec:degroot}
DeGroot proposed a very influential model, closely related to our work, to reason about learning and consensus  in multi-agent systems~\cite{degroot}, in which beliefs are updated by a constant stochastic matrix at each time step. 
%
More specifically, consider 
a group $\{1,2,\ldots,k\}$ of $k$ agents, s.t. each agent
$i$ holds an initial (real-valued) opinion $F_{i}^{0}$ on a given proposition
of interest. 
Let $T_{i,j}$ be a non-negative weight that agent $i$ gives to agent $j$'s opinion, 
s.t. $\sum_{j=1}^k T_{i,j}{=}1$. 
DeGroot's model posits that an agent $i$'s opinion $F_{i}^{t}$ at any time 
$t{\geq}1$ is updated as
$F_{i}^{t}{=}\sum_{j{=}1}^k T_{i,j} F_{i}^{t-1}$.
Letting $F^{t}$ be a vector containing all agents' opinions at time $t$,
the overall update can be computed as $F^{t{+}1}{=}T F^{t}$,
where $T{=}\{T_{i,j}\}$ is a stochastic matrix.
This means that the $t$-th configuration (for $t{\geq}1$) is 
related to the initial one by $F^{t}{=}T^{t}F^{0}$, which is 
a property thoroughly used to derive results in the model.

When we use classical update  
(as 
in Remark~\ref{authority:bias:remark}), our model reduces to DeGroot's 
via the transformation
$F_{i}^{0}{=}\Bfun{i}{0}$, 
and
$T_{i,j}{=}\nicefrac{1}{|\Agents_i|} \ \Ifun{j}{i}$ if $i{\neq}j$,
or 
$T_{i,j}{=}1{-}\nicefrac{1}{|\Agents_i|}\sum_{j{\in}\Agents_i} \Ifun{j}{i}$ otherwise.
Notice that $T_{i,j}{\leq}1$ for all $i$ and $j$, and, by construction, $\sum_{j{=}1}^k T_{i,j}{=}1$ for all $i$. 
The following result is an immediate consequence of this reduction.
\begin{restatable}{corollary}{rescordegroot}
\label{cor:degroot}
    In a strongly connected influence graph $\Ifun{}{}$ and under the classical update function, for all $i, j{\in}\Agents$, $\lim_{t{\to}\infty} \Bfun{i}{t}{=}\lim_{t{\to}\infty} \Bfun{j}{t}$.
\end{restatable}

Unlike its classical counterpart, however, the confirmation-bias update 
(Def.~\ref{def:confirmation-bias}) does not have an immediate correspondence with
DeGroot's model.
Indeed, this update is not linear due the confirmation-bias factor $\CBfun{i}{j}{t}{=}1{-}|\Bfun{j}{t}{-}\Bfun{i}{t}|$.
This means that in our model there is no immediate analogue of the 
relation among arbitrary configurations and the initial one as the 
relation in DeGroot's model (i.e., $F^{t}{=}T^{t}F^{0}$).
Therefore, proof techniques usually used in DeGroot's model
(e.g., based on Markov properties) are not immediately applicable 
to our model.
In this sense our model is an extension of DeGroot's, and we need
to employ different proof techniques to obtain our results.

\section{Conclusions and Other Related Work}
\label{sec:conclusion}
We proposed a model for polarization and belief evolution for multi-agent systems under confirmation-bias.
We showed that whenever all agents can directly or indirectly influence each other, their beliefs always converge, and so does polarization as long as the convergence value is not a 
borderline point.
We also identified necessary conditions for polarization not to disappear, and the convergence value for some important network topologies.
As future work we intend to 
extend our model to model evolution of beliefs and 
measure polarization in situations in which agents
hold opinions about multiple propositions of interest.

\paragraph{Related Work.}
As mentioned in the introduction and discussed in detail in Section \ref{sec:degroot}, the closest related work is on DeGroot models for social learning \cite{degroot}. We summarize some other relevant approaches put into perspective the novelty of our approach. 

\noindent\textbf{Polarization} Polarization was originally studied as a psychological phenomenon in \cite{M76},
and was first rigorously and quantitatively defined by economists Esteban and Ray \cite{Esteban:94:Econometrica}. Their measure of polarization, discussed  in  Section~\ref{sec:model}, is influential, and we adopt it in this paper.
Li et al.\cite{li}, and later Proskurnikov et al. \cite{proskurnikov} modeled consensus and polarization in social networks. Like much other work, they treat polarization simply as the lack of consensus and focus on  when and under what conditions a population reaches consensus.
%
%
Elder's work \cite{alexis} focuses on methods to avoid polarization, without using a quantitative definition of polarization. 
 \cite{Guerra} measures polarization but purely as a function of network topology, rather than taking agents' quantitative beliefs and opinions into account, in agreement with some of our results.

\noindent\textbf{Formal Models} S{\^\i}rbu et al. \cite{sirbu}
use a model  that updates probabilistically to investigate the effects of algorithmic bias on polarization by counting the number of opinion clusters, interpreting a single opinion cluster as consensus.  
Leskovec et al. \cite{gargiulo} simulate social networks and observe group formation over time. 

The Degroot models developed in \cite{degroot} and used in \cite{naive} are closest to ours. Rather than examining polarization and opinions, this work is concerned with the network topology conditions under which agents with noisy data about an objective fact converge to an accurate consensus, close to the true state of the world. As already discussed the basic DeGroot models do not include confirmation bias, however \cite{sikder,mao,mf,hk, robust} all generalize DeGroot-like models to include functions that can be thought of as modelling confirmation bias in different ways, but with either no measure of polarization or a simpler measure than the one we use. \cite{moreau} discusses DeGroot models where the influences change over time, and \cite{survey} presents results about generalizations of these models, concerned more with consensus than with polarization.

\noindent\textbf{Logic-based approaches} Liu et al. \cite{liu} use ideas from doxastic and dynamic epistemic logics to
qualitatively model influence and belief change in social networks. 
%
%
Seligman et al. \cite{fblogic,facebook} introduce a basic ``Facebook logic.'' This logic is non-quantitative, but its interesting point is that an agent's possible worlds are different social networks. This is a promising approach to formal modeling of  epistemic issues in social networks.  Christoff \cite{zoe} extends facebook logic and develops several non-quantitative logics for social networks, concerned with problems related to polarization, such as information cascades.
Young Pederson et al. \cite{myp, myp2, myp3} develop a logic of polarization, in terms of positive and negeative links between agents, rather than in terms of their quantitative beliefs. 
Hunter \cite{hunter} introduces a logic of belief updates over social networks where closer agents in the social network are more trusted and thus more influential. While beliefs in this logic are non-quantitative, there is a quantitative notion of influence 
between users.



\textbf{Other related work} The seminal paper Huberman et al. \cite{huberman} is about determining which friends or followers in a user's network have the most influence on the user. Although this paper does not quantify influence between users, it does address an important question to our project. Similarly, \cite{demarzo} focuses on finding most influential agents. The work on highly influential agents is relevant to our finding that such agents can maintain a network's polarization over time.  

\bibliographystyle{splncs04}
\bibliography{polar}

\version{}{
\newpage
\appendix
\section{Axioms for Esteban-Ray polarization measure}
\label{sec:polar-axioms}

The Esteban-Ray polarization measure used in this paper was developed as the only function (up to constants $\alpha$ and $K$) satisfying all of the following conditions and axioms~\cite{Esteban:94:Econometrica}:
\begin{description}
   \item[Condition H:] 
    The ranking induced by the polarization measure over two
    distributions is invariant w.r.t. the size of the population:~\footnote{This is why we can assume w.l.o.g. that the distribution is a probability distribution.}
    $$
     \PfunER{\pi, y} \geq \PfunER{\pi', y'} \quad \rightarrow \quad \forall \lambda > 0, \,\, \PfunER{\lambda\pi, y} \geq \PfunER{\lambda\pi', y'}~.
    $$
    \item[Axiom 1:] 
    Consider three levels of belief $p, q, r\in[0,1]$ 
    such that the same proportion of the population holds beliefs $q$ and $r$, and
    a significantly higher proportion of the population holds belief $p$.
    If the groups of agents that hold beliefs $q$ and $r$ reach a consensus and
    agree on an \qm{average} belief $\nicefrac{(q+r)}{2}$, then the social network becomes
    more polarized.
    \item[Axiom 2:] 
    Consider three levels of belief $p, q, r \in [0,1]$,
    such that $q$ is at least as close to $r$ as it is to $p$, and
    $p>r$.
    If only small variations on $q$ are permitted, the direction that brings it closer 
    to the nearer and smaller opinion ($r$) should increase polarization.
    \item[Axiom 3:] 
    Consider three levels of belief $p, q, r \in [0,1]$, s.t. $p<q<r$ and there 
    is a non-zero proportion of the population holding belief $q$.
    If the proportion of the population that holds belief $q$ is equally split into
    holding beliefs $q$ and $r$, then polarization increases.
\end{description}

\section{Mathematical details of simulations}
\label{sec:simulations-specification}

In this section we specify the mathematical details of the simulations  in Section~\ref{sec:simulations}.

\subsubsection{Initial belief configurations:}
We consider the following initial belief configurations, depicted in Fig.~\ref{fig:initial-belief-configurations}:

\begin{itemize}
    \item A \emph{uniform} belief configuration representing a set of agents whose beliefs are as varied as possible, all equally spaced in the interval $[0, 1]$:
    \begin{align*}
        \Bfun{i}{0} = \nicefrac{i}{(n{-}1)}~.
    \end{align*}
    
    \item A \emph{mildly polarized} belief configuration with agents evenly split into two groups with moderately dissimilar inter-group beliefs compared to intra-group beliefs: 
    \begin{align*}
        \Bfun{i}{0} = 
        \begin{cases}
            \nicefrac{0.2 i}{\ceil{\nicefrac{n}{2}}} + 0.2, & \text{if $i<\ceil{\nicefrac{n}{2}}$,} \\
            \nicefrac{0.2 (i{-}\ceil{\nicefrac{n}{2}})}{(n{-}\ceil{\nicefrac{n}{2}})} + 0.6 & \text{otherwise.}
        \end{cases}
    \end{align*}
    
    \item An \emph{extremely polarized} belief configuration
    representing a situation in which half of the 
    agents strongly believe the proposition, whereas 
    half strongly disbelieve it:
    \begin{align*}
        \Bfun{i}{0} = 
        \begin{cases}
            \nicefrac{0.2 i}{\ceil{\nicefrac{n}{2}}}, & \text{if $i<\ceil{\nicefrac{n}{2}}$,} \\
            \nicefrac{0.2 (i{-}\ceil{\nicefrac{n}{2}})}{(n{-}\ceil{\nicefrac{n}{2}})} + 0.8, & \text{otherwise.}
        \end{cases}
    \end{align*}
    
    \item A \emph{tripolar} configuration with 
    agents divided into three 
    groups:
    \begin{align*}
        \Bfun{i}{0} =
        \begin{cases}
            \nicefrac{0.2 i}{\lfloor\nicefrac{n}{3}\rfloor}, & \text{if $i < \lfloor{\nicefrac{n}{3}}\rfloor$,} \\
            \nicefrac{0.2 (i{-}\lfloor{\nicefrac{n}{3}}\rfloor)}{(\ceil{\nicefrac{2n}{3}}{-}\lfloor{\nicefrac{n}{3}}\rfloor)} + 0.4, & \text{if $\lfloor{\nicefrac{n}{3}}\rfloor \leq i < \ceil{\nicefrac{2n}{3}}$,} \\
            \nicefrac{0.2 (i{-}\ceil{\nicefrac{2n}{3}})}{(n{-}\ceil{\nicefrac{2n}{3}})} + 0.8, & \text{otherwise.}
        \end{cases}
    \end{align*}
\end{itemize}


\subsubsection{Influence graphs:}
We consider the following influence graphs, depicted in Fig.~\ref{fig:interaction-graphs}:

\begin{itemize}
    \item A \emph{$C$-clique} influence graph $\Interclique$, 
    in which each agent influences every other with constant 
    value $C=0.5$: 
    \begin{align*}
        \Ifunclique{\agent{i}}{\agent{j}} = 0.5~.
    \end{align*}
    This represents the particular case of a social network in
    which all agents interact among themselves, and are all immune
    to authority bias.
    
    \item A \emph{circular} influence graph $\Intercircular$ 
    representing a social network in which agents can be organized
    in a circle in such a way each agent is only influenced by its 
    predecessor and only influences its successor:
    \begin{align*}
        \Ifuncircular{\agent{i}}{\agent{j}} = 
        \begin{cases}
            0.5, & \text{if (i{+}1)\,\text{mod}\,n = j,} \\
            0, & \text{otherwise.}
        \end{cases}
    \end{align*}
    This is a simple instance of a balanced graph (in which
    each agent's influence on others is as high as the influence received, 
    as in Def.~\ref{def:circulation} ahead), which is a pattern commonly 
    encountered in some sub-networks.
    
    \item A \emph{disconnected} influence graph $\Interdisconnected$
    representing a social network sharply divided into two groups in such a way that agents within the same group 
    can considerably influence each other, but not at all
    agents in the other group:
    \begin{align*}
        \Ifundisconnected{\agent{i}}{\agent{j}} = 
        \begin{cases}
            0.5, & \text{if $\agent{i},\agent{j}$ are both
     ${<} \ceil{\nicefrac{n}{2}}$ or both ${\geq} \ceil{\nicefrac{n}{2}}$,} \\
            0, & \text{otherwise.}
        \end{cases}
    \end{align*}

    
    \item An \emph{unrelenting influencers} influence graph $\Interunrelenting$ 
    representing a scenario in which two agents 
    (say, $\agent{0}$ and $\agent{n{-}1}$) 
    exert significantly stronger influence on every other agent than these other agents have among themselves:
    \begin{align*}
        \Ifununrelenting{\agent{i}}{\agent{j}}= 
        \begin{cases}
            0.6, & \text{if $i = 0$ and $j \neq n{-}1$ or $i = n{-}1$
    and $j \neq 0$,} \\
            0, & \text{if $j=0$ or $j=n{-}1$,} \\
            0.1, & \text{if $0\neq i \neq n{-}{1}$ and $0 \neq j \neq n{-}{1}$.}
        \end{cases}
    \end{align*}
    This could represent, e.g., a social network 
    in which two totalitarian media companies dominate the news
    market, both with similarly high levels of influence on 
    all agents.
    The networks have clear agendas to push forward, and are not influenced in a meaningful way by other agents.
\end{itemize}

\section{Proofs}

%

\respolatlimit*

\begin{proof} 
Let  be any real $\epsilon>0$. It suffices to find $N \in \reals$  such that for every $t>N$, $\Pfun{\Blft{t}} < \epsilon.$ 
Let $I_m$ be the bin of $D_k$ such that $v \in I_m.$ Suppose that $l/r$ is the left/right end-point of $I_m$.
 
Take  $\epsilon' =  r$ if $v=0$, $\epsilon' =  l$ if $v=1$ else $\epsilon' = \min\{ v - l, r-v\}.$ Clearly $\epsilon' >0$ because $v$ is not a borderline point. Since $\lim_{t \to \infty} \Bfun{i}{t} = v$, there is $N_i \in \reals$  such that for every $t>N_i$, $|v - \Bfun{i}{t}| < \epsilon'.$  This implies  $\Bfun{i}{t} \in I_m$ for every $t>N_i$.  Take $N= \max\{ N_i | i \in \Agents \}$.  From Prop.~\ref{pol-consensus} $\Pfun{\Blft{t}}=0 < \epsilon$ for every $t>N$  as wanted. \qed
\end{proof}

\reslemmacbmaxdiffmin*

\begin{proof}
We want to prove that $\Bfun{\agent{i}}{t{+}1} \leq \mx{t}$. Since $\Bfun{\agent{j}}{t} \leq \mx{t}$, we can use Def.~\ref{def:confirmation-bias} 
to derive the inequality $\Bfun{\agent{i}}{t{+}1} 
    \leq E_1 \defsymbol \Bfun{\agent{i}}{t} + \frac{1}{|\Agents_i|}\sum_{\agent{j} \in \Agents_i \setminus \{\agent{i}\}}  \CBfun{\agent{i}}{\agent{j}}{t}\Ifun{\agent{j}}{\agent{i}}(\mx{t} -\Bfun{\agent{i}}{t})$. Furthermore, $E_1 \leq E_2 \defsymbol\Bfun{\agent{i}}{t} + \frac{1}{|\Agents_i|}\sum_{\agent{j} \in \Agents_i \setminus \{\agent{i}\}} (\mx{t} -\Bfun{\agent{i}}{t})$ because $\CBfun{\agent{i}}{\agent{j}}{t}\Ifun{\agent{j}}{\agent{i}} \leq 1$ and $\mx{t} - \Bfun{\agent{i}}{t} \geq 0.$ We thus obtain $\Bfun{\agent{i}}{t{+}1} \leq E_2 = \Bfun{\agent{i}}{t} + \frac{|\Agents_i|-1}{|\Agents_i|}(\mx{t} -\Bfun{\agent{i}}{t}) =  \frac{\Bfun{\agent{i}}{t} + (|\Agents_i|-1)\cdot\mx{t}}{|\Agents_i|} \leq \mx{t}$ as wanted. The proof of $\mn{t} \leq \Bfun{\agent{i}}{t{+}1}$ is similar.
\qed
\end{proof}

%

%

%


\begin{restatable}[]{proposition}{respropcbupperboundinequality}
\label{prop:cb-upper-bound-inequality} 
Let $i \in \Agents$, $k \in \Agents_i$, $n,t \in \nat$ with $n\geq 1$, and $v \in [0,1].$
  \begin{enumerate} 
        \item  If $\Bfun{\agent{i}}{t} \leq v$ then 
            $\Bfun{\agent{i}}{t{+}1} \leq v + \nicefrac{1}{|\Agents|}\sum_{\agent{j} \in \Agents_i}\CBfun{i}{j}{t}\Ifun{\agent{j}}{\agent{i}}\left(\Bfun{\agent{j}}{t}-v\right).$ 
        \item $\Bfun{\agent{i}}{t+n} \leq \mx{t} + \nicefrac{1}{|\Agents|} \, \CBfun{i}{k}{t+n-1}\Ifun{\agent{k}}{\agent{i}}(\Bfun{\agent{k}}{t+n-1} - \mx{t}).$ 
     \end{enumerate}
\end{restatable}

\begin{proof} Let $i,k \in \Agents$, $n,t \in \nat$ with $n\geq 1$, and $v \in [0,1].$
\begin{enumerate}
\item From Def.~\ref{def:confirmation-bias}: $\allowbreak \Bfun{\agent{i}}{t{+}1} = \Bfun{\agent{i}}{t} + \frac{1}{|\Agents_i|}\sum_{\agent{j} \in \Agents_i}\CBfun{i}{j}{t}\Ifun{\agent{j}}{\agent{i}}(\Bfun{\agent{j}}{t} - \Bfun{\agent{i}}{t})$ which is less or equal to $v +  \frac{1}{|\Agents_i|}\sum_{\agent{j} \in \Agents_i \setminus \{\agent{i}\}}\CBfun{i}{j}{t}\Ifun{\agent{j}}{\agent{i}}(\Bfun{\agent{j}}{t} - v) \leq v + \frac{1}{|\Agents|}\sum_{\agent{j} \in \Agents_i\setminus \{\agent{i}\}}\CBfun{i}{j}{t}\Ifun{\agent{j}}{\agent{i}}(\Bfun{\agent{j}}{t} - v) = v + \frac{1}{|\Agents|}\sum_{\agent{j} \in \Agents_i}\CBfun{i}{j}{t}\Ifun{\agent{j}}{\agent{i}}(\Bfun{\agent{j}}{t} - v)$ since $|\Agents_i| \leq |\Agents|$.
\item From Prop.~\ref{prop:cb-upper-bound-inequality}(1):  $\Bfun{\agent{i}}{t+n} \leq \mx{t} + \frac{1}{|\Agents|}\sum_{\agent{j} \in \Agents_i}\CBfun{i}{j}{t+n-1}\Ifun{\agent{j}}{\agent{i}}\left(\Bfun{\agent{j}}{t+n-1}-\mx{t}\right) \leq \mx{t} + \frac{1}{|\Agents|}\CBfun{i}{k}{t+n-1}\Ifun{\agent{k}}{\agent{i}}\left(\Bfun{\agent{k}}{t+n-1}-\mx{t}\right)$ using Cor.~\ref{cor:biasfactor} and the fact that $\Bfun{\agent{j}}{t+n-1}-\mx{t} \leq 0$. 
\end{enumerate}
\qed
\end{proof}



\reslemmacbpathbound*

\begin{proof}
\begin{enumerate} 
    \item Let $p$ be the path $\ldinfl{\agent{i_0}}{C_1}{\agent{i_1}}\ldinfl{}{C_2}{}\ldots \ldinfl{}{C_n} {\agent{i_n}}.$ We proceed by induction on $n$. For $n=1$, since $\Bfun{\agent{i_0}}{t} - \mx{t} \leq 0$  we obtain the result immediately from Prop.~\ref{prop:cb-upper-bound-inequality}(2) and Prop.~\ref{fcb-min:prop}.  Assume that $\Bfun{\agent{i_{n-1}}}{t+|p'|} \leq \mx{t} + \frac{C'\CBfunM^{|p'|}}{|\Agents|^{|p'|}}(
    \Bfun{\agent{i_0}}{t} - \mx{t})$  where $p' = \ldinfl{\agent{i_0}}{C_1}{\agent{i_1}}\ldinfl{}{C_2}{}\ldots \ldinfl{}{C_{n-1}} {\agent{i_{n-1}}}$
    and $C' =C_1\times \ldots \times C_{n-1}$  with $n> 1$. Notice that $|p'|=|p|-1$.
    Using Prop.~\ref{prop:cb-upper-bound-inequality}(2), Prop.~\ref{fcb-min:prop}, and the fact that $\Bfun{\agent{i_{n-1}}}{t+|p'|} - \mx{t} \leq 0$ we obtain
    $\Bfun{\agent{i_n}}{t+|p|} \leq \mx{t} + \frac{C_n\CBfunM}{|\Agents|}(\Bfun{\agent{i_{n-1}     }}{t+|p'|} - \mx{t}).$ Using our assumption  
    we obtain $\Bfun{\agent{i_n}}{t+|p|} \leq \mx{t} + \frac{C_n\CBfunM}{|\Agents|}( \mx{t} + \frac{C'\CBfunM^{|p'|}}{|\Agents|^{|p'|}}(
    \Bfun{\agent{i_0}}{t} - \mx{t})  - \mx{t})= \mx{t} + \frac{C\CBfunM^{|p|}}{|\Agents|^{|p|}}(
    \Bfun{\agent{i_0}}{t} - \mx{t})$ as wanted.

    \item Suppose that $p$ is the path $\linfl{\agent{\mstar}^t}{C}{p}{\agent{i}}$. From Lem.~\ref{lemma:cb-path-bound}(1) we obtain  $\Bfun{\agent{i}}{t+|p|} \leq \mx{t} + \frac{C\CBfunM^{|p|}}{|\Agents|^{|p|}}(\Bfun{\agent{\mstar}^t}{t} - \mx{t}) = \mx{t} + \frac{C\CBfunM^{|p|}}{|\Agents|^{|p|}}(\mn{t} - \mx{t}).$  Since $\frac{C\CBfunM^{|p|}}{|\Agents|^{|p|}}(\mn{t} - \mx{t}) \leq 0$, we can substitute $C$ with $\IfunM^{|p|}$. Thus, $\Bfun{\agent{i}}{t+|p|} \leq \mx{t} + \Big(\frac{\IfunM\CBfunM}{|\Agents|}\Big)^{|p|} \allowbreak (\mn{t}-\mx{t}).$  From Th.~\ref{th:cb-max-limits-exist}, the maximum value of $\mn{t}$ is $L$ and the minimum value of $\mx{t}$ is $U$, thus $\Bfun{\agent{i}}{t+|p|} \leq \mx{t} + \Big(\frac{\IfunM\CBfunM}{|\Agents|}\Big)^{|p|}(L-U) = \mx{t} - \delta.$
\end{enumerate}
\qed
\end{proof}

\reslemcbepsilonbound*

\begin{proof}
\begin{enumerate}
    \item Using Prop.~\ref{prop:cb-upper-bound-inequality}(1) with the assumption that $\Bfun{\agent{i}}{t{+}n} \leq \mx{t} - \gamma$ for $\gamma \geq 0$ and the fact that $\Ifun{\agent{j}}{\agent{i}} \in [0,1]$  we obtain the inequality 
    $\Bfun{\agent{i}}{t+n+1} \leq   \mx{t} - \gamma + \frac{1}{|\Agents|}\sum_{\agent{j} \in \Agents_i}\CBfun{i}{j}{t+n}\Ifun{\agent{j}}{\agent{i}}\left(\Bfun{\agent{j}}{t{+}n} - (\mx{t} - \gamma)\right).$  From Cor.~\ref{cor:cb-mbefore-mafter}  $\mx{t}\geq \mx{t+n} \geq \Bfun{\agent{j}}{t+n} $ for every $\agent{j}\in \Agents$, hence $\Bfun{\agent{i}}{t+n+1} \leq \mx{t} - \gamma + \frac{1}{|\Agents|}\sum_{\agent{j} \in \Agents_i}\CBfun{\agent{i}}{\agent{j}}{t+n}\Ifun{\agent{j}}{\agent{i}}\big(\mx{t} - (\mx{t} - \gamma)\big).$ Since $\CBfun{i}{j}{t+n}\Ifun{\agent{j}}{\agent{i}} \in [0,1]$  we derive $\Bfun{\agent{i}}{t+n+1} \leq \mx{t} - \gamma + \frac{1}{|\Agents|}\sum_{\agent{j} \in \Agents_i}\gamma \allowbreak \leq \mx{t} - \frac{\gamma}{|\Agents|}.$
    
    \item Let $p$ be the path $\linfl{\agent{\mstar}^t}{}{p}{\agent{i}}$ where $\mstar^t \in \Agents$ is minimal agent at time $t$ and let $\delta =  \left(\frac{\IfunM\CBfunM}{|\Agents|}\right)^{|p|}(U-L)$. If $|p| = |\Agents|-1$ then the result follows from Lem.~\ref{lemma:cb-path-bound}(2). 
    Else $|p| < |\Agents|-1$ by Def.~\ref{def:influence-path}.  We first show by induction on $m$ that $\Bfun{\agent{i}}{t+|p|+m} \leq \mx{t} - \frac{\delta}{|\Agents|^m}$ for every $m\geq 0$. If $m=0$, $\Bfun{\agent{i}}{t+|p|} \leq \mx{t} - \delta$ by Lem.~\ref{lemma:cb-path-bound}(2). 
    If $m>0$ and $\Bfun{\agent{i}}{t+|p|+(m-1)} \leq \mx{t} - \frac{\delta}{|\Agents|^{m-1}}$ then $\Bfun{\agent{i}}{t+|p|+m} \leq \mx{t} - \frac{\delta}{|\Agents|^{m}}$ by Lem.~\ref{lem:cb-epsilon-bound}(1).  
    Therefore, take $m=|\Agents|-|p|-1$ to obtain  $\Bfun{\agent{i}}{t+|\Agents|-1} \leq \mx{t} - \frac{\delta}{|\Agents|^{|\Agents|-|p|-1}}=\mx{t} - \frac{(\IfunM\CBfunM)^{|p|}.(U-L)}{|\Agents|^{|\Agents|-1}}\leq  \mx{t} - \epsilon$ as wanted.
\end{enumerate}
\qed
\end{proof}



\resthul*

\begin{proof}
Suppose, by contradiction, that $\lim_{t\to\infty} \mx{t}=U \neq L=\lim_{t\to\infty} \mn{t}$.   Let  $\epsilon = \left(\frac{\IfunM\CBfunM}{|\Agents|}\right)^{|\Agents|-1}(U-L).$ From the assumption $U > L$ and Cor.~\ref{cor:biasfactor} we get that $\epsilon > 0$. Take $t=0$ and $m=\left(\ceil{\frac{1}{\epsilon}}+1\right)$. Using  
Cor.~\ref{cor:max-dif} we obtain  $\mx{0}  \geq \mx{{m(|\Agents|-1)}} + m\epsilon.$ Since $m\epsilon > 1$ and $\mx{m(|\Agents|-1)} \geq 0$ then $\mx{0}> 1.$ But this contradicts Def.~\ref{def:extreme:beliefs} which states that $\mx{0} \in [0,1]$.
\qed
\end{proof}

%


\begin{restatable}[Influencing the Extremes]{proposition}{respropcbinfluencingextremes}
\label{prop:cb-influencing-extremes}
If $\Inter$ is strongly connected and $\Blft{0}$ is not radical, then $\mx{|\Agents|{-}1}{<}1$.
\end{restatable}

\begin{proof} Since $\Blft{0}$ is not radical, there must be at least one agent $k$  such that $\Bfun{k}{0}\in (0,1).$  Since $\Inter$ is strongly connected, it suffices to show that for every path $\linfl{\agent{k}}{}{p}{\agent{i}}$ we have $\Bfun{i}{|p|} < 1$. Proceed by induction on size $n$ of the path $p=ki_1\ldots i_{n}$. For $n = 0$, it is true via the hypothesis. For $n \geq 1$, we have, by IH and Def. \ref{def:influence-path}, that $\Bfun{i_{n-1}}{|p|-1} < 1$ and $\Ifun{i_{n-1}}{n} > 0$. Thus, $\Bfun{i_n}{|p|} = \Bfun{i_n}{|p|-1} + \frac{1}{|\Agents_i|}\sum_{\agent{j} \in \Agents_i} \CBfun{i_n}{j}{|p|-1}\Ifun{j}{i_n} (\Bfun{j}{|p|-1} - \Bfun{i_n}{|p|-1})$ separating $i_{n-1}$ from the sum we get $\Bfun{i_n}{|p|} = \Bfun{i_n}{|p|-1} + \frac{1}{|\Agents_i|}\sum_{\agent{j} \in \Agents_i\setminus\{i_{n-1}\}} \CBfun{i_n}{j}{|p|-1}\Ifun{j}{i_n} (\Bfun{j}{|p|-1} - \Bfun{i_n}{|p|-1}) + \frac{1}{|\Agents_i|} \CBfun{i_n}{i_{n-1}}{|p|-1}\allowbreak\Ifun{i_{n-1}}{i_n} (\Bfun{i_{n-1}}{|p|-1} - \Bfun{i_n}{|p|-1}) \allowbreak \leq 1 + \frac{1}{|\Agents_i|}\CBfun{i_n}{i_{n-1}}{|p|-1}\Ifun{i_{n-1}}{i_n} (\Bfun{i_{n-1}}{|p|-1} - 1) < 1$.
\qed
\end{proof}

\restheoremcbgeberalsccconvergence*

\begin{proof}
\begin{enumerate}
    \item If there exists an agent $\agent{k} \in \Agents$ such that $\Bfun{k}{0} \notin \{0,1\}$, we can use Prop.~\ref{prop:cb-influencing-extremes} to show that by the time $|\Agents|-1$ no agent has belief $1$, thus we fall on the general case stated in the beginning of the section (starting at a different time step does not make any difference for this purposes) and, thus, all beliefs converge to the same value according to Cor.~\ref{cor:cb-scc-convergence}.
    \item Otherwise it is easy to see that beliefs remain constant as $0$ or $1$ throughout time, since the agents are so biased that the only agents $\agent{j}$ able to influence another agent $\agent{i}$ ($\CBfun{i}{j}{t} \neq 0$) have the same belief as $\agent{i}$.
\end{enumerate}
\qed
\end{proof}

\respropgroupinfluenceconservation*

\begin{proof} 
Immediate consequence of  Prop. 6.1.1 in \cite{Diestel:17}. 
\qed
\end{proof}

\resprocirculationpath*

\begin{proof}
For the sake of contradiction, assume  that $\Inter$ is balanced (a circulation) and $\Ifun{i}{j} > 0$ but there is no path from $\agent{j}$ to $\agent{i}$. Define the agents reachable from $j$, $R_j = \{k \in 
     \Agents | \ \Path{j}{k} \} \cup \{j\}$ and let $\overline{R}_j = \Agents\setminus R_j$. Notice that $\{ R_j ,\overline{R}_j\}$ is a partition of $\Agents.$ Since  the codomain of $\Inter$ is $[0,1]$, $\agent{i} \in \overline{R}_j$, $\agent{j} \in R_j$ and $\Ifun{i}{j} > 0$ we obtain 
     $\sum_{k \in R_j}\sum_{l \in \overline{R}_j} \Ifun{l}{k} > 0$. Clearly there is no $k \in R_j, l \in \overline{R}_j$ such that $\Ifun{k}{l} > 0$, therefore $\sum_{k \in R_j}\sum_{l \in \overline{R}_j} \Ifun{k}{l} = 0$ which contradicts Prop.~\ref{prop:group-influence-conservation}.
\qed
\end{proof}

\respolnonzero*

\begin{proof} From Lem.~\ref{prop:circulation-path} it follows that  if the influence graph $\Inter$ is balanced and weakly connected then $\Inter$ is also strongly connected. The result follows from  Lem.~\ref{pol-at-limit} and Th.~\ref{theorem:cb-geberal-scc-convergence}. 
\qed
\end{proof}

%

\rescordegroot*

\begin{proof}
    Since the graph is strongly connected it suffices to show that the graph represented by the matrix $P$ in which $P_{i,j}{=}p_{i,j}$ is aperiodic. Since for every individual $i$ $p_{i,i}{>}0$, there is a self-loop, thus no number $K > 1$ divides the length of all cycles in the graph, implying aperiodicity. Thus, the conditions for Theorem 2 of
    \cite{degroot} are met, finishing the proof.
\qed
\end{proof}
}





\end{document}